\documentclass[11pt,notitlepage,tightenlines,superscriptaddress,nofootinbib]{revtex4-1}
\pdfoutput=1

\usepackage[utf8]{inputenc}
\usepackage{amsthm}
\usepackage{amssymb}
\usepackage{amsmath}

\usepackage[colorlinks=true, linkcolor=MidnightBlue, urlcolor=MidnightBlue, citecolor=MidnightBlue, linktoc=section, pdftitle={Computable lower bounds on the entanglement cost of quantum channels}, pdfauthor={Ludovico Lami and Bartosz Regula}]{hyperref}

\usepackage{newpxtext,newpxmath}
\usepackage{bbold}
\usepackage{bbm}
\usepackage{nonfloat}
\usepackage{braket}
\usepackage{dsfont}
\usepackage{mathdots}
\usepackage{mathtools}
\usepackage{enumerate}
\usepackage{csquotes}
\usepackage{stmaryrd}
\usepackage[cal=boondox]{mathalfa}
\usepackage{graphicx}
\usepackage{stackengine}
\usepackage{scalerel}
\usepackage[dvipsnames]{xcolor}
\usepackage{array,xpatch}
\usepackage{makecell}
\newcolumntype{x}[1]{>{\centering\arraybackslash}p{#1}}
\usepackage{tikz}
\usepackage{pgfplots}
\usetikzlibrary{shapes.geometric, shapes.misc, positioning, arrows, decorations.pathreplacing, decorations.text, angles, quotes, backgrounds, calc}
\usepackage{booktabs}
\usepackage{xfrac}
\usepackage{siunitx}
\usepackage{centernot}
\usepackage{comment}
\usepackage{chngcntr}

\newtheorem{thm}{Theorem}
\newtheorem*{thm*}{Theorem}

\newtheorem*{prop*}{Proposition}
\newtheorem{lemma}[thm]{Lemma}
\newtheorem*{lemma*}{Lemma}
\newtheorem{cor}[thm]{Corollary}
\newtheorem*{cor*}{Corollary}

\newtheorem*{cj*}{Conjecture}
\newtheorem{Def}[thm]{Definition}
\newtheorem*{Def*}{Definition}

\makeatletter
\def\thmhead@plain#1#2#3{%
  \thmname{#1}\thmnumber{\@ifnotempty{#1}{ }\@upn{#2}}%
  \thmnote{ {\the\thm@notefont#3}}}
\let\thmhead\thmhead@plain
\makeatother

\theoremstyle{definition}
\newtheorem*{rem}{Remark}
\newtheorem*{note}{Note}

\newcommand{\bb}{\begin{equation}\begin{aligned}}
\newcommand{\bbb}{\begin{equation*}\begin{aligned}}
\newcommand{\ee}{\end{aligned}\end{equation}}
\newcommand{\eee}{\end{aligned}\end{equation*}}
\newcommand\floor[1]{\left\lfloor#1\right\rfloor}
\newcommand\ceil[1]{\left\lceil#1\right\rceil}

\let\textleq\relax
\let\textgeq\relax
\let\texteq\relax

\newcommand{\texteq}[1]{\stackrel{\mathclap{\scriptsize \mbox{#1}}}{=}}
\newcommand{\textleq}[1]{\stackrel{\mathclap{\scriptsize \mbox{#1}}}{\leq}}

\newcommand{\textgeq}[1]{\stackrel{\mathclap{\scriptsize \mbox{#1}}}{\geq}}
\newcommand{\ketbra}[1]{\ket{#1}\!\!\bra{#1}}

\newcommand{\tcb}[1]{{\color{blue} #1}}

\newcommand{\id}{\mathds{1}}
\newcommand{\R}{\mathds{R}}
\newcommand{\N}{\mathds{N}}

\DeclareMathOperator{\Tr}{Tr}

\DeclareMathOperator{\cl}{cl}
\DeclareMathOperator{\co}{conv}
\DeclareMathOperator{\cone}{cone}

\DeclareMathAlphabet{\pazocal}{OMS}{zplm}{m}{n}

\DeclareMathOperator{\supp}{supp}

\newcommand{\HH}{\pazocal{H}}
\newcommand{\T}{\pazocal{T}}
\newcommand{\B}{\pazocal{B}}
\newcommand{\CC}{\pazocal{C}}
\let\C\CC
\newcommand{\Tp}{\pazocal{T}_+}
\newcommand{\D}{\pazocal{D}}
\newcommand{\K}{\pazocal{K}}

\newcommand{\lsmatrix}{\left(\begin{smallmatrix}}
\newcommand{\rsmatrix}{\end{smallmatrix}\right)}

\stackMath
\newcommand\xxrightarrow[2][]{\mathrel{%
  \setbox2=\hbox{\stackon{\scriptstyle #1}{\scriptstyle#2}}%
  \stackunder[4pt]{%
    \xrightarrow{\makebox[\dimexpr\wd2\relax]{$\scriptstyle#2$}}%
  }{%
   \scriptstyle#1\,%
  }%
}}

\newcommand{\tends}[2]{\xxrightarrow[#2]{\mathrm{#1}}}
\newcommand{\tendsn}[1]{\xxrightarrow[\! n\rightarrow \infty\!]{#1}}

\renewcommand{\tends}[2]{\xxrightarrow[#2]{\vphantom{.}\smash{{\raisebox{-1.8pt}{\scriptsize{#1}}}}}}

\stackMath
\def\dyhat{-0.15ex}
\newcommand\myhat[1]{\ThisStyle{%
              \stackon[\dyhat]{\SavedStyle#1}
                              {\SavedStyle\widehat{\phantom{#1}}}}}

\makeatletter
\newcommand*\rel@kern[1]{\kern#1\dimexpr\macc@kerna}
\newcommand*\widebar[1]{%
  \begingroup
  \def\mathaccent##1##2{%
    \rel@kern{0.8}%
    \overline{\rel@kern{-0.8}\macc@nucleus\rel@kern{0.2}}%
    \rel@kern{-0.2}%
  }%
  \macc@depth\@ne
  \let\math@bgroup\@empty \let\math@egroup\macc@set@skewchar
  \mathsurround\z@ \frozen@everymath{\mathgroup\macc@group\relax}%
  \macc@set@skewchar\relax
  \let\mathaccentV\macc@nested@a
  \macc@nested@a\relax111{#1}%
  \endgroup
}

\definecolor{Blues5seq1}{RGB}{239,243,255}
\definecolor{Blues5seq2}{RGB}{189,215,231}
\definecolor{Blues5seq3}{RGB}{107,174,214}
\definecolor{Blues5seq4}{RGB}{49,130,189}
\definecolor{Blues5seq5}{RGB}{8,81,156}

\definecolor{Greens5seq1}{RGB}{237,248,233}
\definecolor{Greens5seq2}{RGB}{186,228,179}
\definecolor{Greens5seq3}{RGB}{116,196,118}
\definecolor{Greens5seq4}{RGB}{49,163,84}
\definecolor{Greens5seq5}{RGB}{0,109,44}

\definecolor{Reds5seq1}{RGB}{254,229,217}
\definecolor{Reds5seq2}{RGB}{252,174,145}
\definecolor{Reds5seq3}{RGB}{251,106,74}
\definecolor{Reds5seq4}{RGB}{222,45,38}
\definecolor{Reds5seq5}{RGB}{165,15,21}

\definecolor{Yellow}{RGB}{233, 233, 150}
\definecolor{Orange}{RGB}{255, 230, 179}
\definecolor{Blueish}{RGB}{192, 242, 217}

\definecolor{goldenyellow}{rgb}{1.0, 0.87, 0.0}

\usetikzlibrary{decorations.text,arrows.meta,bending,shapes.misc}

\tikzset{cross/.style={cross out, draw=black, minimum size=2*(#1-\pgflinewidth), inner sep=0pt, outer sep=0pt},
cross/.default={1pt}}

\pgfdeclarearrow{
  name = mytriangle,
  parameters = { \the\pgfarrowlength, \the\pgfarrowwidth, \the\pgfarrowlinewidth, \ifpgfarrowroundcap c\fi },
  defaults = { length = 0pt .85, width = 0pt 1.3, line width = 0pt .1, round = true },
  setup code = {
    \pgfarrowssettipend{1\pgfarrowlength}
    \pgfarrowssetlineend{.2\pgfarrowlength}
    \pgfarrowssetvisualbackend{0\pgfarrowlength}
    \pgfarrowssetbackend{0\pgfarrowwidth}
    \pgfarrowshullpoint{1\pgfarrowlength}{0\pgfarrowwidth}
    \pgfarrowsupperhullpoint{0\pgfarrowlength}{.5\pgfarrowwidth}
    \pgfarrowssavethe\pgfarrowwidth
    \pgfarrowssavethe\pgfarrowlength
    \pgfarrowssavethe\pgfarrowlinewidth
  },
  drawing code = {
    \ifpgfarrowroundcap\pgfsetroundjoin\fi
    \pgfsetlinewidth{\pgfarrowlinewidth}
    \pgfpathmoveto{\pgfqpoint{1\pgfarrowlength}{0\pgfarrowwidth}}
    \pgfpathlineto{\pgfpoint{0\pgfarrowlength}{.5\pgfarrowwidth}}
    \pgfpathlineto{\pgfqpoint{0\pgfarrowlength}{-.5\pgfarrowwidth}}
    \pgfpathclose
    \pgfusepathqfillstroke
  },
}

\makeatletter
\def\l@subsection#1#2{}
\def\l@subsubsection#1#2{}
\makeatother

\makeatletter
\patchcmd{\@ssect@ltx}
    {\addcontentsline{toc}{#1}{\protect\numberline{}#8}}
    {}
    {}
    {}
\makeatother

\newcommand{\br}[1]{{\color{red!50!yellow!80!black} #1}}

\newcommand{\SEP}{\pazocal{S}}
\newcommand{\PPT}{\pazocal{P\!P\!T}}

\newcommand{\ten}{E^\tau_N}
\newcommand{\tl}{L^\tau}
\newcommand{\tsr}{R^\tau}
\newcommand{\tn}{N_\tau}

\newcommand{\wt}{\widetilde}

\newcommand{\KI}{\mathrm{KE}}
\newcommand{\KIP}{\mathrm{KEP}}

\newcommand{\lset}{\left\{}
\newcommand{\rset}{\right\}}
\renewcommand{\bar}{\,:\,}
\newcommand{\idc}{\mathrm{id}}
\newcommand{\HS}{\pazocal{H\!S}}
\newcommand{\deff}[1]{\textbf{\emph{#1}}}

\let\epsilon\varepsilon
\DeclareMathAlphabet\mathbfcal{OMS}{cmsy}{b}{n}

\let\br\relax
\let\tcb\relax

\begin{document}

\title{Computable lower bounds on the entanglement cost of quantum channels}

\author{Ludovico Lami}
\email{ludovico.lami@gmail.com}
\affiliation{Institut f\"{u}r Theoretische Physik und IQST, Universit\"{a}t Ulm, Albert-Einstein-Allee 11, D-89069 Ulm, Germany}
\affiliation{QuSoft, Science Park 123, 1098 XG Amsterdam, the Netherlands}
\affiliation{Korteweg--de Vries Institute for Mathematics, University of Amsterdam, Science Park 105-107, 1098 XG Amsterdam, the Netherlands}
\affiliation{Institute for Theoretical Physics, University of Amsterdam, Science Park 904, 1098 XH Amsterdam, the Netherlands}

\author{Bartosz Regula}
\email{bartosz.regula@gmail.com}
\affiliation{Department of Physics, Graduate School of Science, The University of Tokyo, Bunkyo-ku, Tokyo 113-0033, Japan}

\begin{abstract}
    A class of lower bounds for the entanglement cost of any quantum state was recently introduced in [\href{https://www.nature.com/articles/s41567-022-01873-9}{Nat.\ Phys.\ 19, 184 (2023)}] in the form of entanglement monotones known as the tempered robustness and tempered negativity. Here we extend their definitions to point-to-point quantum channels, establishing a lower bound for the asymptotic
    entanglement cost of any channel, whether finite or infinite dimensional. This leads, in particular, to a bound that is computable as a semidefinite program and that can outperform previously known lower bounds, including ones based on quantum relative entropy.
    In the course of our proof we establish a useful link between the robustness of entanglement of quantum states and quantum channels, which requires several technical developments such as showing the lower semicontinuity of the robustness of entanglement of a channel in the weak*-operator topology on bounded linear maps between spaces of trace class operators.
\end{abstract}

\maketitle

%%%%%%%%%%%%%%%%%%%%%%%%%%%%%%%%%%%%%%%%%%%%%%%%%%%%%%%%%%%%%%%%%%%%%%%%%%%%%%%%%%%%%%%%%%%%%%

\section{Introduction} \label{intro_channels_sec}

From a modern perspective, one of the fundamental conceptual contributions of classical thermodynamics~\cite{CARNOT, Clausius, Kelvin, PLANCK} is the realisation that information itself is an entity with observable physical consequences. The advent of quantum mechanics and the epistemological revolution it brought~\cite{Planck1901, Heisenberg1925, Born-Jordan, Born-Heisenberg-Jordan} have endowed this statement with a more profound meaning. On the other hand, information theory~\cite{Shannon} has taught us that information can only be understood by means of the operational tasks it enables. In this sense, it is only natural that a key role in the discipline is played by the processes that allow to manipulate information. In the theory of quantum information~\cite{Bennett-distillation, Bennett-distillation-mixed, NC}, which aims to combine quantum physics and information theory, such processes are represented by quantum channels~\cite{Stinespring, Choi}.

Quantum information is concerned, among other things, with how resources can be interconverted~\cite{Church,RT-review}. In this spirit, a great effort has been devoted to the problem of understanding how quantum channels can be transformed into each other. For example, given many uses of a point-to-point quantum channel $\Lambda:A\to B$ connecting Alice's system to Bob's system, one is often interested in determining how much information Alice can transmit to Bob --- the quantum capacity of the channel~\cite{Lloyd-S-D, L-Shor-D, L-S-Devetak, temaconvariazioni}. Equivalently, this problem can be thought of as that of understanding how efficiently one can simulate the noiseless qubit identity channel $\idc_2$ given $\Lambda$.

While this question has received much attention in the last decades, its converse, i.e.\ the problem of determining the rate at which resources are needed to simulate a given (noisy) quantum channel $\Lambda$, although conceptually appealing, is much less studied. Here, the word `resources' can take, depending on the context, different meanings, with the two main ones referring to entanglement and classical communication. For example, a notable result in this area is the quantum reverse Shannon theorem, stating that in the presence of free entanglement the classical communication cost of a quantum channel is equal to its entanglement-assisted capacity, i.e.\ to the amount of classical communication that it could have conveyed in the first place, again in the presence of free entanglement~\cite{Bennett2002, Bennett2014, Berta2011}. While this result is aesthetically pleasing because it establishes a reversible theory, it is not easy to imagine situations in which quantum entanglement, notoriously hard to maintain over long distances, could be counted as a free resource. A complementary approach, instead, is to consider classical communication free and to look primarily at the cost in terms of entanglement consumption. These considerations have inspired the notion of entanglement cost for quantum channels~\cite{Berta2013, Wilde2018}. As is the case for states~\cite{Hayden-EC, Shor2004, Hastings2008}, no single-letter (let alone closed-form) expression is known for the entanglement cost of a quantum channel. In the channel case, the situation is in fact even more intricate than for states, because dynamical resources can be used in a sequential order, where each channel use can influence the subsequent ones.
Therefore, one can identify at least two notions of entanglement cost of a quantum channel, one corresponding to what is needed to simulate parallel instances~\cite{Berta2013}, and the other encompassing possible overheads required for sequential simulation~\cite{Wilde2018}. In fact, even more general schemes for the transformations of quantum channels can be conceived by allowing more exotic protocols that do not assume a fixed causal order of the channel uses~\cite{Chiribella-switch,oreshkov_2012}.

In this work, we focus on the challenging problem of computing lower bounds to the entanglement cost of quantum channels. The fundamental mathematical difficulty associated with this problem is, as for states, the absence of a single-letter formula. In the channel setting, however, further complications linked to the optimisation over entangled input states over many uses of the channel may arise~\cite[Eq.~(1)]{Berta2013}. Here we bypass these difficulties by generalising the tempering method recently introduced by us~\cite{lami_2021-1} to the dynamical setting of quantum channels. Our fundamental result is a semidefinite-programming--computable lower bound on the parallel (and hence also on more general notions, such as the sequential) entanglement cost of a channel in terms of a quantity called the tempered negativity. We show with an example that our new result can improve upon all previously known lower bounds. Building on these findings, we conclude by providing a complete proof of the result announced in~\cite{lami_2021-1} that the theory of point-to-point quantum channel manipulation is fundamentally irreversible under the set of all channel-to-channel transformations that preserve either the set of entanglement-breaking channels, or that of channels with a positive partial transpose.

The rest of the paper is structured as follows. We begin in Section~\ref{sec:prelim} with an introduction to the concepts underlying the investigation of quantum channels and quantum entanglement, and recall the tempering method for quantum states developed in~\cite{lami_2021-1}. Section~\ref{sec:topology} then deals with the problem of how to choose a suitable topology on the space of quantum channels to study their properties: as we show, the most appropriate choice here is an often overlooked weak*-operator topology. Section~\ref{sec:main_bounds} contains the main results of our work: here we rigorously introduce the notions of quantum capacity and quantum channel entanglement cost, generalise the tempering method to quantum channels, and then use the tempered monotones to provide a new lower bound on the entanglement cost of any channel. In Section~\ref{sec:irrev}, we show that the bound can perform better than previously known computable bounds for entanglement cost, and prove the general asymptotic irreversibility of channel manipulation. Our last Section~\ref{sec:robustness_ch_proof} is devoted to a complete proof of one of our technical results used in Section~\ref{sec:main_bounds}, namely the equivalence of the measure of entanglement known as the robustness for states and channels.

\section{Preliminaries}\label{sec:prelim}

\subsection{Quantum systems}

Quantum systems, denoted with capital letters $A$, $B$, etc.\ are mathematically represented by separable\footnote{A Hilbert space, or more generally a Banach space, is said to be separable if it admits a countable norm-dense subset.} Hilbert spaces $\HH_A$, $\HH_B$, and so on. In this paper we shall consider the fully general case of infinite-dimensional spaces, which is arguably the most fundamental --- in fact, all quantum fields that we suspect to model the fundamental constituents of matter are intrinsically infinite-dimensional.

The Banach space of all bounded operators on a Hilbert space $\HH$, equipped with the operator norm $\|X\|_\infty\coloneqq \sup_{\ket{\psi}\in \HH\setminus \{0\}} \frac{\left\|X\ket{\psi}\right\|}{\left\|\ket{\psi}\right\|}$ will be denoted with $\B(\HH)$. We can think of it as the dual of the space of trace class operators on $\HH$ endowed with the trace norm $\|T\|_1\coloneqq \Tr \sqrt{T^\dag T}$, denoted with $\T(\HH)$~\cite[Chapter~VI]{REED}. The duality relation between $\B(\HH)$ and $\T(\HH)$ will be written $\B(\HH) = \T(\HH)^*$. Remarkably, we can in turn think of $\T(\HH)$ as the dual of a Banach sub-space of $\B(\HH)$, that of compact operators on $\HH$, denoted with $\CC(\HH)$, once again equipped with the operator norm. In between $\T(\HH)$ and $\B(\HH)$ lies the Hilbert space of Hilbert--Schmidt operators on $\HH$, called $\HS(\HH)$ equipped with the scalar product $\braket{X,Y}_{\mathrm{HS}}\coloneqq \Tr X^\dag Y$. The fact that $\HS(\HH)$ is actually a Hilbert space rather than simply a Banach space makes it somewhat easier to work with. \tcb{For example, $\HS(\HH)$ can be identified with its own dual.} We summarise the above discussion by stating that the relations
\bb
\tcb{\CC(\HH)^* = \T(\HH) \subseteq \HS(\HH) = \HS(\HH)^* \subseteq \CC(\HH) \subseteq \B(\HH) = \T(\HH)^{*} = \CC(\HH)^{**}}
\ee
hold; here, the inclusions are intended to be between sets (and not Banach spaces), and are all strict unless $\dim\HH<\infty$.

If $\HH$ is separable, as we will always assume, both $\C(\HH)$ and $\T(\HH)$ --- but \emph{not} $\B(\HH)$! --- can be shown to be separable as well (as Banach spaces). The separability of $\T(\HH)$ can be proved simply by taking as a dense subset the set of all operators having a finite expansion with rational coefficients in a fixed orthonormal basis; since separability of the dual space implies separability of the primal~\cite[Theorem~III.7]{REED}, it follows immediately that $\C(\HH)$ is also separable; also, it turns out that $\B(\HH)$ is not separable whenever $\HH$ is infinite dimensional~\cite{inseparability-bounded}.

\subsection{Quantum channels}

Mathematically, a \deff{quantum channel} from a quantum system $A$ to a quantum system $B$, denoted by $\Lambda:A\to B$, is first and foremost a linear map $\Lambda: \T(\HH_A) \to \T(\HH_B)$. In order to be a bona fide quantum channel, $\Lambda$ must satisfy two additional conditions:
\begin{enumerate}[(i)]
\item Complete positivity, which requires that $\idc_n\otimes \Lambda : \T(\C^n \otimes \HH_A) \to \T(\C^n \otimes \HH_B)$ is a positive map for all $n\in \N_+$, where $\idc_n$ denotes the identity map acting on the space of $n\times n$ complex matrices, and positivity of a map $\Gamma$ means that $\Gamma(X)\geq 0$ is positive semidefinite for all positive semidefinite $X\geq 0$.
\item Trace preservation, which requires that the identity $\Tr \Lambda(X) = \Tr X$ is obeyed for all $X\in \T(\HH_A)$.
\end{enumerate}
In what follows, the set of maps satisfying (i)--(ii) will be denoted with $\mathrm{CPTP}_{A\to B}$.

\subsection{Separability and the PPT criterion}

The Hilbert space associated with a bipartite quantum system $AB$ is simply the tensor product of the local spaces, in formula $\HH_{AB} = \HH_A\otimes \HH_B$. A very important set of states within $\D(\HH_{AB})$ is composed of \deff{separable states}, formally defined as the closed convex hull of product states, i.e.
\bb
\SEP^{1}_{AB} \coloneqq \cl\left( \co\left\{ \ketbra{\psi}_A \otimes \ketbra{\phi}_B:\, \ket{\psi}_A\in \HH_A,\, \ket{\phi}_B\in \HH_B,\, \braket{\psi|\psi}=1=\braket{\phi|\phi} \right\} \right) .
\ee
Here, the closure is taken with respect to the trace norm topology (see Section~\ref{sec:topology} for an introduction to topologies for quantum systems). It can be shown~\cite{Holevo2005} that a state $\sigma_{AB}$ is separable if and only if it admits the expression
\begin{equation}
    \sigma_{AB} = \int \ketbra{\psi}_A \otimes \ketbra{\phi}_B\, \mathrm{d}\mu(\psi,\phi)\, ,
    \label{separable}
\end{equation}
where $\mu$ is a Borel probability measure on the product of the sets of local (normalised) pure states. The cone generated inside $\Tp(\HH_{AB})$ by the set of separable states is
\begin{equation}
    \SEP_{AB} \coloneqq \cone \left( \SEP^{1}_{AB} \right) \coloneqq \left\{ \lambda \sigma_{AB}:\, \lambda\geq 0,\, \sigma_{AB}\in \SEP^{1}_{AB}\right\} .
\label{SEP}
\end{equation}

Since deciding whether a state is separable or not is a notoriously intractable problem~\cite{GurvitsNPhard, GharibianNPhard}, some handy criteria have been developed to facilitate this task. The most notable of those is the positive partial transposition (PPT) criterion~\cite{PeresPPT}. The \deff{partial transpose} of some $T_{AB}\in \T(\HH_{AB})$, denoted $T_{AB}^\Gamma$, is defined by first assuming that $T_{AB} = X_A \otimes Y_B$, in which case $T_{AB}^\Gamma = X_A \otimes Y_B^\intercal$, the transposition being with respect to a fixed (but immaterial) basis of $\HH_B$, and then extending the operation to the whole $\T(\HH_{AB})$ by linearity. It is worth observing that the resulting operator $T_{AB}^\Gamma\in \B(\HH_{AB})$ will be bounded but in general not of trace class. We refer the reader to~\cite[Section~VII.A, Supplementary Information]{lami_2021-1} for further details on some subtleties concerning the infinite-dimensional case. We can now observe that the partial transpose of any separable state is necessarily a positive semidefinite operator. Therefore,
\begin{equation}
    \SEP_{AB} \subseteq \PPT_{\!AB} \coloneqq \left\{ T_{AB}\in \Tp(\HH_{AB}):\ T_{AB}^\Gamma\geq 0\right\} ,
\label{PPT}
\end{equation}
which is precisely the aforementioned PPT criterion.

\begin{note}
Hereafter we will denote with $\K_{AB}$ any one of the two cones $\SEP_{AB}$ or $\PPT_{\!AB}$, defined by~\eqref{SEP} and~\eqref{PPT}, respectively. Therefore, a statement involving $\K$ will be intended to hold equally well for $\K=\SEP_{AB}$ or $\K=\PPT_{\!AB}$.
\end{note}

\subsection{Robustness, negativity, and tempering}

From now on all states are implicitly understood to be on a bipartite system $AB$, even though we will often omit the subscripts. Given a state $\rho=\rho_{AB}$, how to quantify its entanglement content? A particularly simple and arguably fruitful idea is to use its \deff{(standard) $\mathbfcal{K}$-robustness}, defined by
\bb
R_\K^s(\rho) \coloneqq&\ \inf \left\{ \Tr \delta:\, \delta \in \K,\, \rho+\delta\in \K \right\} \\
=&\ \frac12 \left( \sup\left\{ \Tr X\rho:\ X\in \left[ -\id, \id\right]_{\K^*} \right\} - 1 \right)
\label{std_rob}
\ee
In the above equation, $\delta$ is assumed to be of trace class, the dual cone $\K^*$ is defined by
\bb
\K^* &\coloneqq \left\{ Z_{AB} = Z_{AB}^\dagger \in \B\left( \HH_{AB}\right):\ \Tr[Z_{AB}W_{AB}]\geq 0\quad \forall\ W_{AB}\in \K_{AB} \right\} \subset \B(\HH_{AB})\, , \label{dual_cone}
\ee
and the corresponding operator interval is
\bb
\left[ -\id, \id\right]_{\K^*} &\coloneqq \left\{ Z_{AB} = Z_{AB}^\dag \in \B\left( \HH_{AB}\right):\ \left|\Tr[Z_{AB}W_{AB}]\right|\leq \Tr W_{AB}\quad \forall\ W_{AB}\in \K_{AB} \right\} . \label{operator_interval}
\ee
Note that $R_\PPT^s(\rho)\leq R_\SEP^s(\rho)$ for all states $\rho$, simply because of the inclusion in~\eqref{PPT}. Let us observe in passing that in this paper we use the original definition of Vidal and Tarrach~\cite{VidalTarrach}, instead of adopting the more recent convention of Ref.~\cite{taming-PRA, taming-PRL}, according to which the robustness would be defined as $1+R_\K$.

A simpler quantity to compute is the \deff{negativity}, defined by~\cite{negativity, plenioprl}
\bb
N(\rho) \coloneqq \left\|\rho^\Gamma\right\|_1 = \sup\left\{ \Tr X\rho:\, X = X^\dagger\!,\ \left\|X^\Gamma\right\|_\infty \leq 1 \right\} .
\label{negativity}
\ee
It is also useful to consider its logarithmic version, the \deff{logarithmic negativity}, defined by
\begin{equation}
    E_N(\rho) \coloneqq \log_2 \left\|\rho^\Gamma\right\|_1 
    %= \log_2 \sup\left\{ \Tr X\rho:\, \left\|X^\Gamma\right\|_\infty\leq 1 \right\} .
    \label{logarithmic_negativity}
\end{equation}
Note that the convention that we are adopting here differs slightly from the one employed by Vidal and Werner~\cite{negativity}, who take the negativity to be $\frac12 \left( N(\rho) - 1\right)$. Our logarithmic negativity, instead, is the same as in~\cite{negativity}.

In~\cite{lami_2021-1}, we introduced a technique called `tempering' that can be applied to yield modified versions of the robustness and the negativity. For a pair of states $\rho,\omega$ on a bipartite system $AB$, the \deff{\textbf{$\boldsymbol{\omega}$-tempered $\mathbfcal{K}$-robustness}} is defined by
\begin{align}
1+2\tsr_{\!\K}(\rho | \omega) &\coloneqq \sup\left\{ \Tr X\rho:\ X\in \left[ -\id, \id\right]_{\K^*},\ \|X\|_\infty=\Tr X\omega \right\} , \label{tsr} \\
\tsr_{\!\K} (\rho ) &\coloneqq \tsr_{\!\K}(\rho|\rho)\, , \label{stsr}
\end{align}
where $\left[ -\id, \id\right]_{\K^*}$ is given by~\eqref{operator_interval} (cf.~\eqref{std_rob}). 
Clearly, this is a convex program, and even a semidefinite one (a.k.a.\ SDP) for the special case $\K=\PPT$.

Analogously, the \deff{$\boldsymbol{\omega}$-tempered negativity} defined by
\begin{align}
\tn(\rho | \omega) &\coloneqq \sup\left\{ \Tr X\rho:\, X = X^\dagger\!,\ \left\|X^\Gamma\right\|_\infty \leq 1,\ \|X\|_\infty=\Tr X\omega \right\} , \label{tn} \\
\tn(\rho) &\coloneqq \tn(\rho|\rho)\, . \label{stn}
\end{align}
Compare this with~\eqref{negativity}. 
Once again, the expression in~\eqref{stn} is in fact an SDP. The corresponding \deff{tempered logarithmic negativity} is
\bb
\ten(\rho) \coloneqq \log_2 \tn(\rho)\, .
\label{ten}
\ee
For a survey of the properties of the tempered robustness and negativity we refer the reader to~\cite[Proposition~S5]{lami_2021-1}. The main application of this quantity in~\cite{lami_2021-1} was to establish a universal and computable lower bound on the entanglement cost of any quantum state ---  the generalisation of this result is precisely the aim of this work.

\section{The unjustly overlooked weak*-operator topology}\label{sec:topology}

Before proceeding with the investigation of entanglement of quantum channels, let us address an issue pertinent to the study of their properties: the choice of a suitable topology on the space of quantum channels.
The purpose of this section is to introduce and discuss the notion of the weak*-operator topology, which will prove instrumental to some of our proofs.
%As we will see, the choice of such a topology will prove instrumental to some of our proofs.

\subsection{Topologies on the set of quantum states}

%Before delving into the world of channels, let us consider topologies at the level of states.
Let us begin by considering topologies at the level of states. The space of interest here is therefore $\T(\HH)$, i.e.\ the Banach space of trace class operators on some separable Hilbert space $\HH$. We will consider mainly two topologies on $\T(\HH)$, namely:
\begin{itemize}
    \item The trace norm topology, induced by the native norm $\|\cdot\|_1$. A sequence $(T_n)_{n\in \N}$ of trace class operators is said to converge with respect to the trace norm topology to some $T\in \T(\HH)$, and we write $T_n\tends{tn}{n\to\infty} T$, if $\left\|T_n - T\right\|_1 \tends{}{n\to\infty} 0$.
    \item The weak* topology induced by the duality $\T(\HH) = \C(\HH)^*$, where $\C(\HH)$ denotes the space of compact operators on $\HH$. Equivalently, it can be defined as the coarsest topology that makes all functionals of the form $T \mapsto \Tr TK$, where $K\in \C(\HH)$, continuous. Accordingly, a sequence\footnote{Strictly speaking, since the weak* topology is non-metrisable in general, we should be talking about nets rather than sequences. However, we will see that this technical complication can be avoided in most cases of interest here.} $(T_n)_{n\in \N}$ in $\T(\HH)$ will be said to converge to $T\in \T(\HH)$ with respect to the weak* topology, denoted $T_n\tends{w*}{n\to\infty} T$, if $\Tr T_n K \tends{}{n\to\infty} \Tr TK$ for all compact $K\in \C(\HH)$. 
\end{itemize}
Clearly, the weak* topology is coarser than the trace norm topology, which implies that any sequence that converges with respect to the former topology converges (to the same limit) also with respect to the latter.

Trace norm and weak* topology are however genuinely different in infinite dimension. For example, picking an orthonormal basis $\{\ket{n}\}_{n\in \N}$ of $\HH$, one sees that the sequence $\left( \ketbra{n}\right)_{n\in \N}$ has no limit with respect to the former topology, yet it satisfies $\ketbra{n} \tendsn{w*} 0$.

A word of caution before proceeding is advisable: the weak* topology considered here is \emph{not} the same commonly used in the von~Neumann algebra approach to quantum theory~\cite[Remark~2]{achievability}.

The fundamental reason why the weak* topology is so useful to us lies in the Banach--Alaoglu theorem, which states that \emph{for every Banach space $X$, the unit ball of $X^*$ is weak*-compact}~\cite[Theorem~IV.21]{REED}. In our operator setting, by applying this result to the duality $\T(\HH) = \C(\HH)^*$ we immediately deduce the following:

\begin{lemma} \label{Banach-Alaoglu_trace_class}
The unit ball $B_1\coloneqq \{T\in \T(\HH) : \|T\|_1\leq 1\}$ is weak*-compact.
\end{lemma}

\subsection{Topologies on the set of quantum channels}

We now move on to the discussion of topologies on spaces of quantum channels. For the sake of this presentation, let us fix two quantum systems $A$ and $B$, and let us use the shorthand notation $\T_A\coloneqq \T(\HH_A)$ and $\T_B\coloneqq \T(\HH_B)$ for the respective spaces of trace class operators. Quantum channels can be thought of as elements of the Banach space $\B\left(\T_A\to \T_B \right)$ of linear maps $\Lambda:\T_A\to \T_B$ that are bounded with respect to the trace norm, i.e.\ that satisfy $\|\Lambda\|_{1\to 1} \coloneqq \sup_{X\in \T_A,\, \|X\|_1\leq 1} \left\|\Lambda(X)\right\|_1<\infty$. We can turn $\B\left(\T_A\to \T_B \right)$ into a Banach space by equipping it with the norm $\|\cdot\|_{1\to 1}$. As it turns out, every completely positive and trace preserving map, and hence any quantum channel, belongs to $\B\left(\T_A\to \T_B \right)$; moreover, its norm is precisely $1$.

When it comes to the choice of a topology on $\B\left(\T_A\to \T_B \right)$, there are several possibilities. The most common choices are however essentially two:
\begin{itemize}
    \item The diamond norm topology, induced by \tcb{a norm alternative to $\|\cdot\|_{1\to 1}$ called the diamond norm (completely bounded trace norm). This is given by~\cite{kitaev_1997}
    \begin{equation}\begin{aligned}
    \|\Lambda\|_\Diamond \coloneqq \sup_{\rho} \left\| \left[\idc \otimes \Lambda\right](\rho) \right\|_1,
    \label{diamond}
    \end{aligned}\end{equation}
    with the optimisation being over all bipartite quantum states $\rho \in \D(\HH_A \otimes \HH_{A})$ on two copies of the Hilbert space of Alice's system\footnote{The diamond norm is sometimes defined through an optimisation over only pure state $\psi \in \D(\HH_A \otimes \HH_A)$ or over all states $\rho \in \D(\HH_C \otimes \HH_A)$ with system $C$ arbitrary; all such notions are equivalent~\cite{gilchrist_2005}.}. This distance represents a natural extension of the trace distance to quantum channels, obeying an equivalent of the \br{Helstrom--Holevo} theorem: the diamond norm distance between any two quantum channels captures the difficulty in distinguishing them operationally~\cite{Sacchi2005}.

    According to the diamond norm} topology, a sequence $(\Lambda_n)_{n\in \N}$ in $\B\left(\T_A\to \T_B \right)$ is said to converge to $\Lambda\in \B\left(\T_A\to \T_B \right)$ if $\left\|\Lambda_n - \Lambda\right\|_\diamond \tendsn{} 0$. This is essentially the choice made in the definitions~\eqref{ch_distillable}--\eqref{ch_cost}.
    
    \item Since the above topology turns out to be too strong for many purposes\tcb{, most notably when dealing with infinite-dimensional systems~\cite{PLOB, Shirokov2018, VV-diamond},} it is customary to employ also the strong operator topology, induced by the family of semi-norms $\Lambda \mapsto \left\|\Lambda(X)\right\|_1$, for all $X\in \T_A$. This implies that a sequence of channels $\left(\Lambda_n\right)_{n\in \N}$ in $\B\left(\T_A\to \T_B \right)$ converges to $\Lambda$ with respect to the strong operator topology, and we write $\Lambda_n\tendsn{so} \Lambda$, if and only if $\left\|\Lambda_n(X) - \Lambda(X) \right\|_1\tendsn{} 0$ for all $X\in \T_A$.
\end{itemize}
Despite the name, the strong topology is actually weaker (i.e.\ coarser) than the diamond norm topology. And still, for what we have in mind it is too strong. In order to exploit the power of the Banach--Alaoglu theorem, we need to devise a version of the weak* topology that applies to the channel setting. The simple solution is the following.

\begin{Def}
The \emph{weak*-operator topology} on $\B\left(\T_A\to \T_B \right)$ is defined as the coarsest topology that makes all functionals $\Lambda \mapsto \Tr\left[ Y \Lambda(X) \right]$ continuous, where $X\in \T_A$ and $Y\in \CC_B$, and $\CC_B$ denotes the Banach space of compact operators on $\HH_B$ equipped with the operator norm.
\end{Def}

We immediately see that a sequence $(\Lambda_n)_{n\in\N}$ in $\B\left(\T_A\to \T_B \right)$ converges to $\Lambda\in \B\left(\T_A\to \T_B \right)$ with respect to the weak*-operator topology, which we will write $\Lambda_n \tendsn{w*o} \Lambda$, if and only if $\Lambda_n(X) \tendsn{w*} \Lambda(X)$ for all $X\in \T_A$.

\begin{rem}
The use of the weak* topology in the context of quantum resource theories of states was explored in~\cite{taming-PRA, taming-PRL, nonclassicality, achievability}. Building on that, the possibility of extending these concepts to spaces of quantum channels was considered in~\cite{Haapasalo2021}. Related topologies, such as the bounded weak topology of~\cite{arveson_1969,paulsen_2002}, appeared in the literature before, but they have been employed in a rather different way --- as is customary in the von~Neumann algebra community, they were defined for sets of unital maps, which can be considered as adjoints of quantum channels, acting on operators in $\B(\HH)$ rather than $\T(\HH)$ (i.e.\ the Heisenberg picture).
\end{rem}

With these tools at hand, we can now obtain the following.

\begin{lemma} \label{Xi_weak*op_compact_lemma}
The unit ball of $\B\left(\T_A\to \T_B \right)$, i.e.\ the set
\begin{equation}
\Xi \coloneqq \left\{ \Lambda:\T_A\to \T_B:\, \left\|\Lambda(X)\right\|_1\leq \|X\|_1\quad \forall\, X\in \T_A \right\} ,
\label{Xi}
\end{equation}
is compact with respect to the weak*-operator topology.
\end{lemma}

\begin{proof}
In order to apply the Banach--Alaoglu theorem we need to identify the weak*-operator topology with the weak* topology induced on $\B\left(\T_A\to \T_B \right)$ by a pre-dual. Construct the vector space
\begin{equation}
    \T_A\otimes \CC_B \coloneqq \left\{ \sum_{i=1}^N X_i\otimes Y_i:\, N\in \N,\, X_i\in \T_A,\, Y_i\in \CC_B\right\} .
\end{equation}
We can turn it into a normed space by defining the \emph{projective tensor norm} on it through the expression~\cite[p.~27]{DEFANT}
\begin{equation}
    \left\|Z\right\|_\pi \coloneqq \inf \left\{ \sum_{i=1}^N \|X_i\|_1 \|Y_i\|_\infty\, :\ Z = \sum_{i=1}^N X_i\otimes Y_i,\, N\in \N \right\} .
    \label{pi_norm}
\end{equation}
Let $\T_A\myhat{\otimes}_\pi \CC_B$ denote the completion of $\T_A\otimes \CC_B$ with respect to the norm $\|\cdot\|_\pi$. It is well known that~\cite[p.~27]{DEFANT}
\begin{equation}
    \left( \T_A\myhat{\otimes}_\pi \CC_B \right)^* = \B\left(\T_A\to \CC_B^* \right) = \B\left(\T_A \to \T_B \right) ,
\end{equation}
with the duality taking the form
\begin{equation}
    \braket{\Lambda, X\otimes Y} \coloneqq \Tr\left[ Y\,\Lambda(X)\right]
    \label{duality}
\end{equation}
for all $\Lambda \in \B\left(\T_A \to \T_B \right)$, $X\in \T_A$, and $Y\in \CC_B$, and extended by linearity and continuity to the whole $\T_A\myhat{\otimes}_\pi \CC_B$. 

Hence, $\T_A\myhat{\otimes}_\pi \CC_B$ is a pre-dual of $\B\left(\T_A \to \T_B \right)$. Then, the Banach--Alaoglu theorem~\cite[Theorem~IV.21]{REED} tells us that the dual unit ball $\Xi$ is compact in the weak* topology induced by $\T_A\myhat{\otimes}_\pi \CC_B$. According to this topology, a net $(\Lambda_\alpha)_\alpha$ in $\B\left(\T_A \to \T_B \right)$ converges to $\Lambda\in \B\left(\T_A \to \T_B \right)$ if and only if $\braket{\Lambda_\alpha, Z}\tends{}{\alpha} \braket{\Lambda,Z}$ for all $Z\in \T_A\myhat{\otimes}_\pi \CC_B$. By choosing $Z$ to be of the form $Z=X\otimes Y$, we see that the weak* topology we just defined is finer than the weak*-operator topology discussed above. Hence, since $\Xi$ is compact with respect to the former topology, it must be such with respect to the latter as well.\footnote{In fact, these two topologies can be shown to coincide on $\Xi$, by virtue of the general fact that a subset of a Banach space and its norm closure generate the same weak* topology on any bounded subset of the dual space.}
\end{proof}

\section{Bounding the entanglement cost of quantum channels}\label{sec:main_bounds}

\subsection{Quantum capacity and entanglement cost of a channel}

We are now interested in the study of quantum communication, in which the manipulated objects are quantum channels themselves. To this end, we specify the relevant sets of channels which can be regarded as having no entanglement, and hence, as basically useless for the purposes of transmitting quantum systems. These should be thought of as the equivalent of separable and PPT states at the level of maps. Recalling that $\K_{AB}$ denotes either one of the cones $\SEP_{AB}$ and $\PPT_{\!AB}$, let us define the set of \emph{\textbf{$\mathbfcal{K}$-enforcing channels}}, denoted $\KI$, as
\begin{equation}\begin{aligned}
    \KI_{A\to B} \coloneqq \left\{ \Gamma \in \mathrm{CPTP}_{A\to B} \,:\, \left[\idc \otimes \Gamma \right] (X) \in \K_{AB} \quad \forall\ X \in \T_+\big(\HH_A \otimes \HH_A \big) \right\}.
\end{aligned}\end{equation}
In particular, any normalised quantum state $\rho$ satisfies $[\idc_k \otimes \Gamma](\rho) \in \K^1_{AB}$ when $\Gamma \in \KI$, with $\K^1_{AB}$ standing for the set of operators in $\K_{AB}$ with trace one.
Without loss of generality, leveraging the fact that $\K$ is weak*-closed (in fact, it would suffice to have it trace norm closed), one can constrain the ancillary space that the identity channel is acting on to be finite-dimensional~\cite{Holevo2005}, in the sense that
\begin{equation}\begin{aligned}
    \KI_{A\to B} = \left\{ \Gamma \in \mathrm{CPTP}_{A\to B} \,:\, \left[\idc_k \otimes \Gamma \right] (X) \in \K_{RB} \quad \forall\ X \in \T_+\big(\HH_R \otimes \HH_A \big), \ \forall\ \HH_R \cong \mathds{C}^k,\  k\in \N \right\}.
\end{aligned}\end{equation}
When $\K = \pazocal{S}$, the $\K$-enforcing channels are known as \emph{entanglement breaking}~\cite{HSR}; when $\K = \PPT$, they correspond to \emph{PPT-binding} (or entanglement-binding) channels~\cite{horodecki_2000-1}. In finite-dimensional spaces, these are precisely the channels whose Choi-Jamiołkowski states $[\idc \otimes \Gamma](\Phi_{d_A})$ are separable or PPT, respectively~\cite{HSR, horodecki_2000-1}.

In the context of entanglement theory, the manipulation of quantum %channels
resources is typically realised using the class of local operations and classical communication (LOCC)~\cite{Bennett-error-correction, Berta2013}, which consists of all protocols where the two communicating parties can perform arbitrary channels on the local parts of their systems, and communicate classical information (e.g.\ measurement results) to each other. However, many fruitful bounds and relations have been obtained by relaxing the considered set of processes, allowing the two parties to employ larger classes of protocols~\cite{Rains2001, leung_2015, wang_2019-3, berta_2017, Kaur2017, Xin-exact-PPT, regula_2020-2}. 
In order to understand the ultimate capabilities of such channel manipulation schemes, and indeed also to avoid the ambiguity in choosing a `right' type of transformations to consider, we instead follow the axiomatic approach of~\cite{BrandaoPlenio1, BrandaoPlenio2, lami_2021-1} inspired by ideas that first emerged in the context of thermodynamics~\cite{Caratheodory1909, GILES, Lieb-Yngvason}, and set out to establish a bound that would apply to \emph{all} relevant processes, without making assumptions about their structure.

As the basic axiom of any communication scheme, we assume that a valid channel manipulation protocol should transform $\K$-enforcing channels without generating any additional resources. 
We consider this to be the weakest constraint that any physical communication protocol that could reasonably be deemed as `free', i.e.\ effectively inexpensive to implement, should satisfy.

To understand why, it is instructive to look at a setting that violates our assumption, such as that of the reverse Shannon theorem. In this framework, pre-shared entanglement is provided for free to the parties, and the only costly resource is instead classical communication~\cite{Bennett2014, Berta2011}. Clearly, by adding entanglement it is possible to transform an entanglement-breaking channel into something that is not entanglement-breaking, which contradicts our assumption. 
As mentioned in the Introduction, this setting is interesting because it leads to a reversibility result: the classical communication cost of implementing the channel is the same as the amount of classical communication that can be extracted from it~\cite{Bennett2002, Bennett2014, Berta2011}. And yet, it is not entirely clear in what concrete setting entanglement could be considered to be a cheaper resource than classical communication. The present state of affairs, in which we have serious difficulties establishing entanglement over distances larger than a thousand kilometers~\cite{entanglement-1200-km} but we routinely communicate classically with the Voyager~1 probe, more than 7 orders of magnitude more distant~\cite{Voyager-1}, casts some doubts on the practicality of this route. Our axioms, instead, do not incur this problem, as they prohibit the creation of entanglement for free altogether. They can thus be thought of as a possible extension of the LOCC paradigm that, although much more permissive than the standard LOCC framework, treats entanglement as a costly resource at all stages of the protocol.

We thus define, first at the level of transformations of single channels, the set of \textbf{\emph{$\mathrm{\boldsymbol{KE}}$-preserving quantum processes}}:
\begin{equation}\begin{aligned}
    \KIP_{[A \to B] \to [A' \to B']} \coloneqq \lset \Upsilon_{[A \to B] \to [A' \to B']} \bar \Upsilon(\mathrm{CPTP}_{A \to B}) \subseteq \mathrm{CPTP}_{A' \to B'},\; \Upsilon(\KI_{A \to B}) \subseteq \KI_{A' \to B'} \rset.
\end{aligned}\end{equation}
We do not assume any specific structure of such protocols; although physical transformations of channels are typically taken to have the form of so-called quantum superchannels~\cite{Chiribella2008}, here we do not need to presuppose that. We can also write $\KIP_{[A \to B]^{\otimes n} \to [A' \to B']^{\otimes m}}$ to denote processes which act on $n$ parallel copies of the space of maps from system $A$ to system $B$, in the sense that $\Lambda^{\otimes n}_{A \to B} \in [A \to B]^{\otimes n}$.

However, there are more general ways in which transformation protocols could access $n$ uses of a given channel~\cite{chiribella_2009,Chiribella-switch,oreshkov_2012}. We use the notation $[A \to B]^{\times n}$ to denote $n$-tuples of maps from $A$ to $B$, representing arbitrary uses of multiple channels. That is, a process which uses $n$ channels $(\Lambda_1, \ldots, \Lambda_n) \in [A \to B]^{\times n}$ does not have to use them in parallel, but can use them in any physically consistent manner, including transformations which do not have a fixed causal order. A more general form of a $\KI$-preserving quantum process can then be defined as
\begin{equation}\begin{aligned}
    \KIP_{[A \to B]^{\times n} \to [A' \to B']} \coloneqq \big\{ \Upsilon_{[A \to B]^{\times n} \to [A' \to B']} \,:\, &\Upsilon(\Lambda_1, \ldots, \Lambda_n) \in \mathrm{CPTP}_{A' \to B'} \; \forall\ \Lambda_i \in \mathrm{CPTP}_{A \to B},\\ &\Upsilon(\Gamma_1, \ldots, \Gamma_n) \in \KI_{A' \to B'} \; \forall\ \Gamma_i \in \KI_{A \to B} \big\}.
\end{aligned}\end{equation}
We note that each transformation $\Upsilon \in \KIP_{[A \to B]^{\times n} \to [A' \to B']}$ is assumed to be an $n$-linear map.

The quantum capacity $Q(\Lambda)$ is then the maximum rate $R$ at which $\K$-enforcing $n$-channel processes can simulate the noiseless communication channel $\idc_2^{\otimes \ceil{Rn}}$ when the channel $\Lambda$ is used $n$ times. The (parallel) entanglement cost $E_C(\Lambda)$, on the other hand, is given by the least rate at which noiseless identity channels $\idc_2$ are required in order to simulate the action of $n$ parallel copies of the given noisy channel $\Lambda$. Precisely,
\begin{align}
    Q_{\KIP}(\Lambda_{A \to B}) &\coloneqq \sup\left\{R>0:\, \lim_{n\to\infty} \, \inf_{\Upsilon_n \in \KIP\left([A \to B]^{\times n} \to [A_0 \to B_0]^{\otimes \ceil{Rn}}\right)} \left\| \Upsilon_n\!\left( \Lambda_{A \to B}^{\times n} \right) - \idc_2^{\otimes \ceil{Rn}} \right\|_\Diamond = 0 \right\} , \label{ch_distillable}\\[.5ex]
    E_{C,\,\KIP}(\Lambda_{A \to B}) &\coloneqq \inf\left\{R>0:\, \lim_{n\to\infty} \, \inf_{\Upsilon_n \in \KIP\left([A_0 \to B_0]^{\otimes \floor{Rn}} \to [A \to B]^{\otimes n}\right)} \left\| \Upsilon_n\!\left( \idc_2^{\otimes \floor{Rn}} \right) - \Lambda_{A \to B}^{\otimes n} \right\|_\Diamond = 0 \right\} . \label{ch_cost}
\end{align}
The distance used in the above \tcb{definitions is the operationally meaningful diamond norm, given by~\eqref{diamond}~\cite{kitaev_1997,Sacchi2005}.}
%is based on the diamond norm (completely bounded trace norm)~\cite{kitaev_1997}, and is given by
%\begin{equation}\begin{aligned}
%  \|\Lambda - \Gamma\|_\Diamond \coloneqq \sup_{\rho} \left\| \left[\idc_d \otimes \Lambda\right](\rho) - \left[\idc_d \otimes \Gamma\right](\rho) \right\|_1,
%  \label{diamond}
%\end{aligned}\end{equation}
%with the optimisation being over all bipartite quantum states $\rho \in \D(\HH_A \otimes \HH_{A})$ on two copies of the Hilbert space of Alice's system\footnote{The diamond norm is sometimes defined through an optimisation over only pure state $\psi \in \D(\HH_A \otimes \HH_A)$ or over all states $\rho \in \D(\HH_C \otimes \HH_A)$ with system $C$ arbitrary; all such notions are equivalent~\cite{gilchrist_2005}.}. This distance represents a natural extension of the trace distance to quantum channels, obeying an equivalent of the Holevo--Helstrom theorem: when the diamond norm distance between two channels vanishes, there is no physical process which can distinguish them better than a random guess~\cite{Sacchi2005}.

Let us remark about the different systems in play in~\eqref{ch_distillable}--\eqref{ch_cost}. In the definition of $Q_{\KIP}(\Lambda_{A \to B})$, we write $\Upsilon_n\!\left( \Lambda_{A \to B}^{\times n} \right)$ to denote that the copies of the channel $\Lambda$ do not have to be provided as a tensor product $\Lambda^{\otimes n}$, but can be used in any desired way. The target of this protocol is the channel $\idc_2^{\otimes \ceil{Rn}}$, representing $\ceil{Rn}$ qubits of noiseless quantum communication. In contrast, in our definition of $E_{C,\,\KIP}(\Lambda_{A \to B})$ we use a certain number of identity channels to simulate the action of $\Lambda^{\otimes n}$ in parallel. Importantly, this is \emph{not} the most general definition of channel entanglement cost, and indeed more general simulation schemes can be considered~\cite{Wilde2018}. However, the important point for us is that this definition is the \textit{lowest} possible entanglement cost of $\Lambda$ --- having to simulate $\Lambda^{\otimes n}$ is easier than having to simulate $n$ arbitrary uses of it, so our definition of $E_{C,\,\KIP}$ lower bounds more general ones~\cite{Wilde2018}. 

We also stress that the choice of the broad class of $\KI$-preserving quantum processes means that other choices of manipulation process --- and in particular LOCC --- are necessarily subsets of $\KIP$. This immediately gives $Q_{\KIP}(\Lambda) \geq Q_{\mathrm{LOCC}}(\Lambda)$ and $E_{C,\KI}(\Lambda) \leq E_{C,\mathrm{LOCC}}(\Lambda)$.

Finally, we can also consider an extension of the definitions in~\eqref{ch_distillable}--\eqref{ch_cost} which incorporates a non-zero transformation error --- that is, we no longer demand that the transformation be asymptotically exact, but only that the final error do not exceed a given threshold. The resulting modified notions of quantum capacity and channel entanglement cost are
\begin{align}
    Q^\epsilon_{\KIP}(\Lambda_{A \to B}) &\coloneqq \sup\left\{R\!>\!0:\, \limsup_{n\to\infty} \inf_{\Upsilon_n \in \KIP\left([A \to B]^{\times n} \to [A_0 \to B_0]^{\otimes \ceil{Rn}}\right)} \frac12 \left\| \Upsilon_n\!\left( \Lambda_{A \to B}^{\times n} \right) - \idc_2^{\otimes \ceil{\!Rn\!}} \right\|_\Diamond \!\leq \epsilon \right\}\! , \label{ch_distillable_eps}\\[.5ex]
    E^\epsilon_{C,\,\KIP}(\Lambda_{A \to B}) &\coloneqq \inf\left\{R\!>\!0:\, \limsup_{n\to\infty} \inf_{\Upsilon_n \in \KIP\left([A_0 \to B_0]^{\otimes \floor{Rn}} \to [A \to B]^{\otimes n}\right)} \frac12\left\| \Upsilon_n\!\left( \idc_2^{\otimes \floor{\!Rn\!}} \right) - \Lambda_{A \to B}^{\otimes n} \right\|_\Diamond \!\leq \epsilon \right\}\! . \label{ch_cost_eps}
\end{align}

\subsection{Robustness of entanglement of a channel}

Mirroring the quantification of resources such as entanglement for quantum states, one can ask about how to effectively measure the resource content of a channel. Although such concepts date back to the early days of quantum information~\cite{Bennett2003}, it was not until recently that resource measures of quantum channels were formalised~\cite{Berta2013,takeoka_2014,PLOB,liu_2020,liu_2019-1,gour_2019-1,Gour2019,bauml_2019}. One such measure can be naturally defined by extending the concept of robustness measures~\cite{VidalTarrach}, which we encountered in the definition of $R^s_\K$.
The \deff{(standard) $\mathbfcal{K}$-enforcing robustness} is
\begin{equation}\begin{aligned}
    R^s_{\KI} ( \Lambda ) = \inf \left\{ \lambda \,:\, \Lambda + \lambda \Gamma \in (1+\lambda) \KI,\; \Gamma \in \KI  \right\}.
\end{aligned}\end{equation}
Note that this quantity is in general \textit{not} just the robustness of entanglement of the corresponding Choi state~\cite{yuan_2021}. 

The crucial property of the robustness is its monotonicity under all $\KI$-preserving quantum processes, which we show explicitly for completeness.
\begin{lemma}\label{lem:rob_ch_monotonicity}
Let $\Lambda: A \to B$ be any channel. Then, for any $\Upsilon \in \KIP_{[A \to B]\to[A' \to B']}$ it holds that
\begin{equation}\begin{aligned}
 R^s_{\KI} ( \Lambda ) \geq R^s_{\KI} \left(\Upsilon(\Lambda)\right).
\end{aligned}\end{equation}
\end{lemma}
\begin{proof}
Let $\Lambda + \lambda \Gamma \in (1+\lambda) \KI$, $\Gamma \in \KI$ be any feasible decomposition of $\Lambda$. Then $\Upsilon(\Gamma) \in \KI$ by definition of $\KIP$, and similarly $\Upsilon(\Lambda) + \lambda \Upsilon(\Gamma) \in (1+\lambda) \KI$. Optimising over all feasible decompositions gives the statement of the Lemma.
\end{proof}
The robustness $R^s_{\KI}$ is defined at the level of channels, rather than states, which prevents a direct application of methods established for the state case, such as those in Ref.~\cite{lami_2021-1}.
However, we will show this quantity to obey a very strong relation with the state-based robustness measure $R^s_\K$:  the channel robustness can be computed by optimising the state robustness $R^s_\K$ over all input states.

\begin{lemma}[(Channel--state equivalence of the robustness)]\label{lem:rob_psi_ineq}
For any channel $\Lambda: A \to B$, it holds that
\begin{equation}\begin{aligned}
    R^s_{\KI} ( \Lambda ) = R^s_{\K} ( \Lambda ) \coloneqq \sup_{\rho} R^s_\K\left(\left[\idc \otimes \Lambda\right] (\rho)\right),
\end{aligned}\end{equation}
where the maximisation is over all states $\rho \in \D(\HH_A \otimes \HH_A)$ (or, equivalently, over all pure states $\psi$, or over states $\rho \in \D(\HH_R \otimes \HH_A)$ with $R$ arbitrary).
\end{lemma}
The proof of this Lemma is one of the main technical contributions of this work. Due to its length, we defer it to Sec.~\ref{sec:robustness_ch_proof}.

Our main idea will be therefore to go from quantities defined at the level of channels to quantities defined at the level of states, which will allow us to extend the reasoning of Ref.~\cite{lami_2021-1} to channel manipulation. In particular, we define the \emph{\textbf{channel tempered robustness}} and \emph{\textbf{channel tempered negativity}} as
\begin{align}
    %1+2\tsr_{\!\K}(\Lambda | \Theta) &\coloneqq \sup_{\rho} \tsr_{\!\K}([\idc \otimes \Lambda](\rho)\, |\, [\idc \otimes \Theta](\rho))\\
    \tsr_{\!\K}(\Lambda | \Theta) &\coloneqq \sup_{\rho} \tsr_{\!\K}\left([\idc \otimes \Lambda](\rho)\, |\, [\idc \otimes \Theta](\rho)\right)\\
\tsr_{\!\K} (\Lambda ) &\coloneqq \tsr_{\!\K}(\Lambda | \Lambda)\\
    \tn(\Lambda | \Theta) &\coloneqq \sup_{\rho} \br{\tn} \left([\idc \otimes \Lambda](\rho)\,|\, [\idc \otimes \Theta](\rho)\right)\\
\tn (\Lambda ) &\coloneqq \br{\tn}(\Lambda | \Lambda),\\
\ten(\Lambda) &\coloneqq \log_2 \tn(\Lambda).
\end{align}

The final auxiliary result that we will need gives the precise value of the channel-based robustness $ R^s_{\KI}$ for the identity channel.
\begin{lemma}\label{lem:idrob}
It holds that $R^s_{\KI}(\idc_d) = d-1$.
\end{lemma}
\begin{proof}
Recall from~\cite{VidalTarrach} that a maximally entangled state can be written as $\Phi_d = d \sigma_+ - (d-1) \sigma_-$, where
\begin{equation}\begin{aligned}
    \sigma_+ \coloneqq \frac{\id + d\Phi_d}{d(d+1)}\, ,\qquad \sigma_- \coloneqq \frac{\id-\Phi_d}{d^2-1}
\end{aligned}\end{equation}
are both separable states~\cite{Horodecki1999}. It is not difficult to notice that $\sigma_\pm$ correspond to the Choi states of valid quantum channels, i.e.\ $\sigma_\pm = [\idc\otimes\Gamma_\pm](\Phi_d)$ for some $\Gamma_\pm$, which means that such channels are entanglement breaking~\cite{HSR}. This gives a valid feasible decomposition for $\idc_d$ as $\idc_d = d \Gamma_+ - (d-1) \Gamma_-$, implying that $R^s_\KI(\idc_d) \leq d-1$. On the other hand, it is known that the state-based robustness $R^s_\K$ satisfies $R^s_{\K}(\Phi_d) = d-1$~\cite{VidalTarrach}, which by Lemma~\ref{lem:rob_psi_ineq} gives
\begin{equation}\begin{aligned}
    R^s_{\KI}(\idc_d) \geq R^s_{\K}([\idc_d\otimes\idc_d](\Phi_d)) = d-1
\end{aligned}\end{equation}
for both the separable and the PPT cone, so equality must hold.
\end{proof}

\begin{rem}
The fact that the robustness $R^s_\KI(\Lambda)$ of a channel equals the state-based robustness $R^s_\K(J_{\Lambda})$ of the corresponding Choi-Jamiołkowski state is a more general property satisfied by so-called teleportation-simulable channels~\cite{Bennett-error-correction,horodecki_1999-1,muller_thesis,PLOB,Kaur2017}. Here we limited ourselves to the channel $\idc_d$ for simplicity.
\end{rem}

\subsection{General bounds on \texorpdfstring{$\boldsymbol{E_C}$}{E\_C}}

With the definitions in place, we are ready to state and prove the main technical result of this paper.

\begin{thm} \label{main_tool_thm_ch}
With $\KI$ denoting either entanglement-breaking or PPT-binding channels, the entanglement cost under $\KI$-preserving quantum processes satisfies that
\bb\label{eq:main_first}
\inf_{\epsilon\, \in\, [0,\,1/2)} E_{C,\,\KIP}^\epsilon (\br{\Lambda}) \geq \tl_{\!\K}(\Lambda)\, ,
\ee
where
\bb\label{eq:main_second}
\tl_{\!\K}(\br{\Lambda}) \coloneqq \limsup_{n\to\infty} \frac1n \log_2 \left( 1+\tsr_{\!\K}\left(\Lambda^{\otimes n}\right) \right) \geq \ten(\Lambda) \br{{}\geq \ten(J_\Lambda)}.
\ee
\end{thm}

\br{
\begin{rem}
The evaluation of the left-hand side of \eqref{eq:main_first} is, in general, very difficult due to two obstacles: one is the optimisation in the transformation error $\epsilon$, and the other is the computation of the limit $n\to\infty$ of the number of channel copies $n$ which is used to define $E_{C,\,\KIP}^\epsilon$ (Eq.~\eqref{ch_cost_eps}). Our first bound in~\eqref{eq:main_first} alleviates the former problem, giving in particular a general lower bound on the entanglement cost $E_{C,\,\KIP}$. The resulting bound, however, still requires an asymptotic regularisation. The crucial aspect of our second bound in~\eqref{eq:main_second} is that it is single letter --- no optimisation over many channel copies is needed, and the quantity corresponds to an optimisation problem of a fixed size. For any input state $\rho$, the computation of $E^\tau_N([\idc \otimes \Lambda](\rho))$ is a semidefinite program~\cite{lami_2021-1}, making the bound efficiently computable in practice.
\end{rem}
}

\begin{proof}
The argument follows closely that given in the proof of Theorem~S7 in Ref.~\cite{lami_2021-1}. Let $R$ be an achievable rate for the entanglement cost $E_{C,\, \KIP}^\epsilon (\Lambda)$ at some error threshold $\epsilon\in [0,1/2)$, as per the definition in~\eqref{ch_cost_eps}. Consider a sequence of operations $\Upsilon_n\in \KIP\left([A_0 \to B_0]^{\otimes \floor{Rn}} \to [A \to B]^{\otimes n}\right)$, with $A_0,B_0$ single-qubit systems, such that
\bb
\epsilon_n\coloneqq \frac12\left\| \Upsilon_n\!\left( \idc_2^{\otimes \floor{Rn}} \right) - \Lambda_{A \to B}^{\otimes n} \right\|_\Diamond
\label{main_tool_proof_eq1_ch}
\ee
with
\bb
\limsup_{n\to\infty} \epsilon_n\leq \epsilon < \frac12\, .
\label{main_tool_proof_eq2_ch}
\ee
For all sufficiently large $n$, we then write
\bb
2^{\floor{Rn}} &\texteq{(i)} 1 + R_\KI^s\left( \idc_2^{\otimes \floor{Rn}}\right) \\
&\textgeq{(ii)} 1 + R_\KI^s\left( \Upsilon_n\left(\br{\idc}_2^{\otimes \floor{Rn}}\right)\right) \\
&\texteq{(iii)} 1 + R_\K^s\left( \Upsilon_n\left(\br{\idc}_2^{\otimes \floor{Rn}}\right)\right) \\
&\textgeq{(iv)} (1-2\epsilon_n) \left( 1+\tsr_{\!\K}\left(\Lambda^{\otimes n}\right) \right) + \epsilon_n \\
&\geq (1-2\epsilon_n) \left( 1+\tsr_{\!\K}\left(\Lambda^{\otimes n}\right) \right) \\
&\textgeq{(v)} (1-2\epsilon_n) \frac{\tn\left(\Lambda^{\otimes n}\right) + 1}{2} \\
&\textgeq{(vi)} (1-2\epsilon_n) \frac{\tn\left(\Lambda\right)^n + 1}{2} \\
&\geq \frac{1-2\epsilon_n}{2}\, \tn(\Lambda)^n\, .
\label{chain_ch}
\ee
Let us go through each of the steps in detail.
\begin{enumerate}[(i)]
    \item Follows by Lemma~\ref{lem:idrob}, where we established the exact value of $R^s_\KI$ for the identity channel.
    \item Follows from the monotonicity of $R_\KI^s$ under all $\KI$-preserving quantum processes (Lemma~\ref{lem:rob_ch_monotonicity}).
    \item Here we go from the channel-based quantity $R^s_\KI$ to the state-based quantity $R^s_\K$ through an application of the channel--state equivalence (Lemma~\ref{lem:rob_psi_ineq}).
    \item This is the channel analogue of Lemma~S6 of \cite{lami_2021-1}, which can be seen as follows. Suppose that two channels satisfy $\frac12 \|\Lambda-\Lambda'\|_{\Diamond} \leq \epsilon$. Then
\begin{equation}\begin{aligned}
& 1+2\tsr_{\!\K}(\Lambda' | \Lambda) \\
&\quad = \sup_{\rho}\left\{ \Tr\! \left(X\, [\idc\otimes\Lambda'](\rho)\right)\,:\, X\in \left[ -\id, \id\right]_{\K^*},\ \|X\|_\infty=\Tr\! \left(X\, [\idc\otimes\Lambda](\rho)\right) \right\} \\
&\quad \geq \sup_\rho\Big\{ \left(1 - \left\| [\idc\otimes\Lambda'](\rho) - [\idc\otimes\Lambda](\rho)\right\|_{1}\right) \\
&\hphantom{\quad \geq \sup_\rho\Big\{}\times \Tr\! \left(X\, [\idc\otimes\Lambda](\rho)\right) :\ X\in \left[ -\id, \id\right]_{\K^*},\ \|X\|_\infty=\Tr\! \left(X\, [\idc\otimes\Lambda](\rho)\right) \Big\} \\
&\quad \geq \inf_\rho \left\{1 - \left\| [\idc\otimes\Lambda'](\rho) - [\idc\otimes\Lambda](\rho)\right\|_{1} \right\} \\ &\qquad \times \sup_\rho\left\{ \Tr\! \left(X\, [\idc\otimes\Lambda](\rho)\right) :\ X\in \left[ -\id, \id\right]_{\K^*},\ \|X\|_\infty=\Tr\! \left(X\, [\idc\otimes\Lambda](\rho)\right) \right\} \\
&\quad = \left( 1 - \left\|\Lambda' - \Lambda\right\|_{\Diamond} \right) \left( 1+ 2\tsr_{\!\K}(\Lambda)\right) \\
&\quad \geq (1-2\epsilon) \left( 1+ 2\tsr_{\!\K}(\Lambda)\right).
\end{aligned}\end{equation}
\item Follows by Proposition~S5(d) of \cite{lami_2021-1}, which tells us that $\tsr_{\! \SEP}(\rho) \geq \tsr_{\! \PPT}(\rho) \geq \frac12 (\tn(\rho)-1)$ for any state $\rho$.
    \item This step is a consequence of the super-multiplicativity of the channel tempered negativity; explicitly,
\bb
\tn(\Lambda^{\otimes n}) &= \sup_{\rho \in \D\left(\HH_A^{\otimes n} \otimes \HH_A^{\otimes n}\right)} \tn\!\left([\idc^{\otimes n} \otimes \Lambda^{\otimes n}](\rho)\right)\\
    &\geq \sup_{\rho \in \D(\HH_A \otimes \HH_A)} \tn\!\left([\idc^{\otimes n} \otimes \Lambda^{\otimes n}](\rho^{\otimes n})\right)\\
    &= \sup_{\rho \in \D(\HH_A \otimes \HH_A)} \tn\!\left( \left[[\idc \otimes \Lambda](\rho)\right]^{\otimes n} \right)\\
    &\geq \sup_{\rho \in \D(\HH_A \otimes \HH_A)} \tn\!\left([\idc \otimes \Lambda](\rho)\right)^n\\
    &= \tn(\Lambda)^n,
\ee
where in the second inequality (on the fourth line) we used the supermultiplicativity of $N_\tau$ for states, i.e.\ the fact that $N_\tau(\rho^{\otimes n}) \geq N_\tau(\rho)^n$, which was shown in~\cite[Proposition~S5(e)]{lami_2021-1}.
\end{enumerate}
Let us now go back to~\eqref{chain_ch}. Applying the logarithm, dividing by $n$, and taking the limit superior as $n \to \infty$ concludes the proof. The stated inequality with the quantity $\tl_{\!\K}(\Lambda)$ follows by applying this procedure to the inequality in step~(iv).
\end{proof}

Let us stress again that the $\KI$-preserving processes considered here are a larger class than typically employed ones, such as LOCC or PPT processes. Since the entanglement cost can only increase when a smaller type of channel manipulation schemes is used, the bound of Theorem~\ref{main_tool_thm_ch} applies also to any smaller class, and in particular for any channel $\Lambda$ it holds that
\begin{equation}\begin{aligned}
  E_{C,\mathrm{LOCC}}(\Lambda) \geq \ten(\Lambda).
\end{aligned}\end{equation}
We will shortly see that this bound can outperform previously known ones.

\tcb{
\begin{rem}
    We observe in passing that the same reasoning used to derive the bounds appearing in Theorem~\ref{main_tool_thm_ch} in terms of tempered quantities, which ultimately relies on the properties of the partial transpose, can be repeated for another operation known as \emph{reshuffling} (or realignment)~\cite{Resh1,Resh2,Resh3}. The outcome is another family of lower bounds for the channel entanglement cost, possibly independent of the one provided here. Indeed, the underlying reasoning can be extended also beyond the resource theory of entanglement. A complete account of these developments will be published soon~\cite{insights}.    
\end{rem}
}

\section{Irreversibility of channel manipulation: a detailed proof}\label{sec:irrev}

In~\cite{lami_2021-1} we presented a result establishing the fundamental irreversibility under non-entangling operations of the theory of entanglement manipulation for states, as well as its extension to the channel setting~\cite[Methods]{lami_2021-1}. The purpose of this section is to provide a complete proof of this result, leveraging the technical tools honed in the previous section. We start by recalling the definition of the two-qutrit state $\omega_3$, whose irreversibility under general non-entangling protocols was shown in~\cite{lami_2021-1}; it is defined by
\bb
\omega_3 \coloneqq \frac12 \left( P_3 - \Phi_3 \right) .
\label{omega_3}
\ee
Here, $P_3\coloneqq \sum_{j=1}^3 \ketbra{jj}$ is the projector onto the maximally correlated subspace, and $\Phi_3 =\ketbra{\Phi_3}$, with $\ket{\Phi_3}\coloneqq \frac{1}{\sqrt3}\sum_i \ket{ii}$, is the maximally entangled state of dimension 3. We then considered the qutrit-to-qutrit channel $\Omega_3$ whose Choi state is $\omega_3$. This is given by~\cite[Methods]{lami_2021-1}
\begin{equation}\begin{aligned}\label{omega_3_ch}
    \Omega_3 \coloneqq \frac{3}{2} \Delta - \frac{1}{2} \idc_3,
\end{aligned}\end{equation}
where $\Delta(\cdot) = \sum_{i,j=1}^3 \ketbra{i} \cdot \ketbra{i}$ is the dephasing channel, setting to $0$ all off-diagonal elements of the input matrix. %The channel $\Omega_3$ is used in~\cite{lami_2021-1} to provide an example where the entanglement cost under KE-preserving operations exceeds the corresponding quantum capacity.
The fact that the entanglement cost of the channel $\Omega_3$ exceeds its corresponding quantum capacity, even under all KE-preserving transformations, has been announced in~\cite{lami_2021-1}, as we now recall.
%We use this ansatz to establish a channel equivalent of Theorem~1 in~\cite{lami_2021-1}, showing the irreversibility of quantum communication in the general setting considered in our work.

\begin{thm}[{\cite{lami_2021-1}}] \label{irreversibility_thm_ch}
The qutrit-to-qutrit channel $\Omega_3$ defined by~\eqref{omega_3_ch} satisfies that
\bb
Q_{\KIP}(\Omega_3) \leq \log_2 \frac32 \approx 0.585
\label{omega_3_distillable_ch}
\ee
but
\bb
E_{C,\, \KIP}^\epsilon (\Omega_3) \geq 1
\label{omega_3_cost_ch}
\ee
for all $\epsilon\in [0,\, 1/2)$. In particular, the resource theory of communication is irreversible under quantum processes which preserve either the set of entanglement-breaking channels or that of PPT-binding channels.
\end{thm}

The above result establishes the fundamental irreversibility of the theory of manipulation of point-to-point channels, and it is therefore in direct analogy with the other main findings of~\cite{lami_2021-1} concerning the theory of bipartite states. In the above setting, we consider as free all those protocols that in some sense do not introduce additional entanglement into the system, a philosophy encapsulated in our choice of free operations as KE-preserving processes. It was already known~\cite{Wilde2018} that the theory of quantum channel manipulation is irreversible when only local operations and classical communication are allowed, while here we extend this to the much broader class of KEP transformations.

As we have mentioned previously, a classic result of quantum information known as the reverse quantum Shannon theorem~\cite{Bennett2002, Bennett2014, Berta2011} establishes instead the \emph{reversibility} of the theory under different circumstances, i.e.\ when unlimited entanglement is given for free, and instead it is classical communication that is deemed a costly resource. Since it is easy to see that this latter approach does not comply with our assumptions, our results and the reverse quantum Shannon theorem are not in direct contradiction and instead complement each other. Indeed, Theorem~\ref{irreversibility_thm_ch} can be thought of as a general no-go result: when no entanglement creation is allowed, the irreversibility of quantum communication cannot be circumvented even by going beyond LOCC. The ability to generate entanglement is therefore necessary to achieve reversible channel transformations and establish an equivalent of the reverse Shannon theorem.

\begin{proof}[Proof of Theorem~\ref{irreversibility_thm_ch}]
The lower bound on $E_C$ follows from Theorem~\ref{main_tool_thm_ch}: we have that
\begin{equation}\begin{aligned}
    E_{C,\, \KIP}^\epsilon (\Omega_3) &\geq \ten(\Omega_3)\\
&= \sup_{\rho} \ten([\idc \otimes \Omega_3](\rho))\\
&\geq \ten([\idc \otimes \Omega_3](\Phi_3))\\
&=\ten(\omega_3)\\
&= 1,
\end{aligned}\end{equation}
where the last equality was shown in Theorem~S9 of~\cite{lami_2021-1}.

To bound the quantum capacity of $\Omega_3$, we will use the channel divergence based on the max-relative entropy $D_{\max}$~\cite{Datta2009,Berta2013} (also known as generalised robustness). This bound first appeared in Ref.~\cite{christandl_2017} for transformation protocols $\Upsilon$ restricted to adaptive LOCC quantum combs. Recently, it was shown in Ref.~\cite{regula_2020-2} that the max-relative entropy in fact provides a strong converse bound on quantum capacity assisted by general, $\KI$-preserving quantum processes --- specifically, it holds that
\begin{equation}\begin{aligned}
    Q_{\KIP}(\Omega_3) \leq \inf_{\Gamma \in \KI} D_{\max} (\Omega_3 \| \Gamma),
\end{aligned}\end{equation}
where
\begin{equation}\begin{aligned}\label{eq:dmax_ch}
    D_{\max}(\Lambda \| \Gamma) &= \log_2 \inf \lset 1+\lambda \bar \Lambda + \lambda \Xi = (1+\lambda) \Gamma,\; \Xi \in \mathrm{CPTP} \rset.
\end{aligned}\end{equation}
Since the completely dephasing channel $\Delta$ is explicitly entanglement breaking, we get
\begin{equation}\begin{aligned}
   Q_{\KIP}(\Omega_3) &\leq D_{\max} (\Omega_3 \| \Delta)\\
   &\leq \log_2 \frac{3}{2},
\end{aligned}\end{equation}
where we used the ansatz $\Omega_3 + \frac{1}{2} \idc_{3} = \frac{3}{2} \Delta$ as a feasible solution for~\eqref{eq:dmax_ch}.
\end{proof}

Previous lower bounds on the entanglement cost fall broadly into two categories. The first one is quantities that require complicated optimisation and are typically intractable in practice, such as the regularised relative entropy of entanglement $E_r^\infty$~\cite{Vedral1997} or the squashed entanglement~\cite{Tucci1999, squashed, faithful, Shirokov-sq}. The second type are \emph{computable} measures, which can be efficiently evaluated. The latter category includes the measured relative entropy of entanglement~\cite{Piani2009} or the SDP lower bound of Ref.~\cite{irreversibility-PPT}. Importantly, to date, all of the computable bounds were in fact lower bounds on the regularised relative entropy of entanglement, and thus they can never perform better than the bound obtained using $E_r^\infty$. Our bound based on $\ten$, on the other hand, can be strictly better: since the quantum relative entropy is upper bounded by $D_{\max}$~\cite{Datta2009}, we have that
\begin{equation}\begin{aligned}
 E_{C,\, \KIP}^\epsilon (\Omega_3) \geq \ten(\Omega_3) > \log_2 \frac32 &\geq \inf_{\Gamma \in \KI} D_{\max}(\Omega_3 \| \Gamma) \\
  &= \inf_{\Gamma \in \KI} \sup_{\rho} D_{\max} (\idc \otimes \Omega_3(\rho) \| \idc \otimes \Gamma(\rho))\\
  &\geq \sup_{\rho} E^\infty_{r,\K} (\idc \otimes \Omega_3(\rho)),
\end{aligned}\end{equation}
where the \br{equality in the} second line was shown in~\br{\cite[Lemma~12]{Berta2018}}, and
\begin{equation}\begin{aligned}
  E_{r,\K}^\infty (\omega) \coloneqq \lim_{n\to\infty} \frac{1}{n} \inf_{\sigma \in \K^1_{A^nB^n}} \! D(\omega^{\otimes  n}\|\sigma)
\end{aligned}\end{equation}
with $D(\omega \| \sigma) = \Tr \omega (\log_2 \omega - \log_2 \sigma)$.
This shows that the tempered negativity bound --- itself efficiently computable as a semidefinite program --- can not only outperform all other computable bounds, but even the regularised relative entropy bound.

\section{Proof of Lemma~\ref{lem:rob_psi_ineq}}\label{sec:robustness_ch_proof}

Let us restate the result used before for the reader's convenience.

{
\renewcommand{\thethm}{\ref{lem:rob_psi_ineq}}
\begin{lemma}\label{lem:rob_psi_ineq2}
For any channel $\Lambda: A \to B$, it holds that
\begin{equation}\begin{aligned}
    R^s_{\KI} ( \Lambda ) = R^s_{\K} ( \Lambda ) \coloneqq \sup_{\rho} R^s_\K\left(\left[\idc \otimes \Lambda\right] (\rho)\right),
\end{aligned}\end{equation}
where the maximisation is over all states $\rho \in \D(\HH_A \otimes \HH_A)$ (or, equivalently, over all pure states $\psi$, or over states $\rho \in \D(\HH_R \otimes \HH_A)$ with $R$ arbitrary).
\end{lemma}
}

We begin with a helpful lemma that will allow us to recast the robustness as an optimisation over sub-normalised quantum operations. Specifically, we define a \emph{\textbf{$\mathbfcal{K}$-enforcing subchannel}} to be a completely positive map $\Gamma$ which satisfies $[\idc_k \otimes \Gamma](X) \in \K_{RB}$ for all $k\in \N$ and $X \geq 0$, and which is also trace non-increasing, in the sense that $\Tr \Gamma(\rho) \leq 1$ for all density operators $\rho$. We denote the set of all $\K$-enforcing subchannels with $\wt\KI$. We then have the following.

\begin{lemma}\label{lem:channel_rob_subnormalised}
The robustness $R^s_{\KI}$ can be equivalently expressed by optimising over $\K$-enforcing subchannels. Specifically, for any positive and trace preserving map $\Lambda$ it holds that
\begin{equation}\begin{aligned}
    R^s_{\KI} ( \Lambda  ) = R^s_{\wt{\KI}} ( \Lambda ) \coloneqq \inf \left\{ \lambda \,:\, \Lambda + \lambda \Gamma \in (1+\lambda) \wt{\KI},\; \Gamma \in \wt{\KI} \right\}.
\end{aligned}\end{equation}
\end{lemma}

\begin{proof}
First, notice that when $\Lambda + \lambda \Gamma = (1+\lambda) \Theta$ for trace-preserving maps $\Lambda$ and $\Theta$, in the non-trivial case where $\lambda>0$, also $\Gamma$ is automatically constrained to be trace preserving. We thus write
\begin{equation}\begin{aligned}
    R^s_{\KI} ( \Lambda ) = \inf \left\{ \lambda \,:\, \Lambda + \lambda \Gamma \in (1+\lambda) \KI,\; \Gamma \in \wt{\KI} \right\}
\end{aligned}\end{equation}
without loss of generality. Clearly, $R^s_{\KI} ( \Lambda  ) \geq R^s_{\wt{\KI}} ( \Lambda  )$ as the latter minimises over a larger set. Consider then a feasible solution for $R^s_{\wt{\KI}}$ of the form $\Lambda + \lambda \Gamma = (1+\lambda) \Theta$ where $\Gamma, \Theta \in \wt{\KI}$. Define the maps
\begin{equation}\begin{aligned}
    \Theta' (X) \coloneqq \Theta(X) + \left[\Tr X - \Tr \Theta(X)\right] \sigma\\
    \Gamma' (X) \coloneqq \Gamma(X) + \left[\Tr X - \Tr \Gamma(X)\right] \sigma
\end{aligned}\end{equation}
for a fixed state $\sigma \in \T(\HH_B)$. 
Now, $X \mapsto \sigma \,\Tr X$ is a $\K$-enforcing map, which means in particular that $\Gamma', \Theta' \in \KI$ by the convexity of $\K$. But then $\Lambda + \lambda \Gamma' = (1+\lambda) \Theta'$, so $R^s_{\KI} \leq \lambda$. Since this holds for arbitrary feasible $\lambda$, we get $R^s_{\KI} ( \Lambda  ) = R^s_{\wt{\KI}} ( \Lambda  )$ as desired.
\end{proof}

The next ingredient we need is the compactness of $\wt\KI$ with respect to an appropriate topology.

\begin{cor} \label{KE_tilde_compact_cor}
For $\K=\SEP, \PPT$, the set $\wt\KI$ of $\K$-enforcing subchannels is compact with respect to the weak*-operator topology.
\end{cor}

\begin{proof}
First, the cone of positive maps inside $\B\left(\T_A\to \T_B\right)$ is weak*-operator closed. To see this, consider a net\footnote{A net on a set $\pazocal{X}$ is simply a function $f:\pazocal{A}\to \pazocal{X}$, where $\pazocal{A}$ is an arbitrary directed set, i.e.\ a set equipped with a pre-order $\leq$ such that given any two elements $a,b\in \pazocal{A}$ one can find a common upper bound $a,b\leq c\in \pazocal{A}$. In this context, we need to use nets rather than simple sequences because the weak*-operator topology is `non-metrisable', i.e.\ it is not induced by any metric, unless $\dim\HH<\infty$.} of positive maps $(\Lambda_\alpha)_\alpha$ converging to $\Lambda$ in the weak*-operator topology, where $\Lambda_\alpha,\Lambda\in \B\left(\T_A\to \T_B\right)$. For all $X\in \T\left(\HH_A\right)$, $X\geq 0$, and all $\ket{\psi} \in \HH_B$, we have that
$\braket{\psi|\Lambda(X)|\psi} = \lim_\alpha \braket{\psi|\Lambda_\alpha(X)|\psi}\geq 0$, where we computed the limit thanks to the fact that $\ketbra{\psi}$ is a compact operator. Hence, $\Lambda(X)\geq 0$; since this holds for all $X\geq 0$, we deduce that $\Lambda$ is positive, as claimed.

Secondly, also the cone of completely positive map is weak*-operator closed. In fact, with the above notation, $\Lambda_\alpha \tends{w*o}{\alpha} \Lambda$ implies that $\idc_k\otimes \Lambda_\alpha \tends{w*o}{\alpha} \idc_k\otimes \Lambda$ for all $k\in \N$, because tensoring with a finite-dimensional space cannot affect weak*-operator convergence. Since $\idc_k\otimes \Lambda_\alpha$ is positive for all $\alpha$, by the above result so is $\idc_k \otimes \Lambda$. This ensures that $\Lambda$ is completely positive.

Thirdly, it is straightforward to verify that the cone $\cone(\KI)\coloneqq \left\{\mu \Lambda:\, \mu\geq 0,\, \Lambda\in \KI\right\}$ of $\K$-enforcing maps is weak*-operator closed as well. This follows from a similar reasoning as above  --- recalling from Ref.~\cite{Holevo2005} that it suffices to verify that $\left[\idc_k \otimes \Lambda\right](X) \in \K$ for all $X \in \Tp$ for finite-dimensional ancillary spaces --- together with the fact that $\K$ itself is weak*-closed.

Finally, we can write
\begin{equation}
    \wt{\KI} = \Xi \cap \cone(\KI)\, ,
    \label{KE_tilde_intersection}
\end{equation}
where $\Xi$ is the unit ball of the space $\B\left(\T_A\to\T_B\right)$, defined by~\eqref{Xi}. To see why, notice that $\left\|\Lambda\right\|_{1\to 1} = \sup_\rho \Tr \Lambda(\rho)$ whenever $\Lambda$ is positive (in particular, when it is completely positive), so that for completely positive maps $\left\|\Lambda\right\|_{1\to 1}\leq 1$ amounts to $\sup_\rho \Tr \Lambda(\rho)\leq 1$. Having established~\eqref{KE_tilde_intersection}, we deduce that %$\Xi$,
$\wt{\KI}$, being an intersection of a weak*-operator compact (cf.\ Lemma~\ref{Xi_weak*op_compact_lemma}) and a weak*-operator closed set, is itself weak*-operator compact.
\end{proof}

Following the techniques of Refs.~\cite{taming-PRL, taming-PRA}, we now show that the above result implies the lower semicontinuity of the channel robustness $R^s_\KI$.

\begin{lemma} \label{channel_robustness_lsc_lemma}
The channel robustness $R^s_\KI$ is lower semicontinuous with respect to the weak*-operator topology, in the sense that $\Lambda_\alpha \tends{\emph{w*o}}{\alpha} \Lambda$ for a net $(\Lambda_\alpha)_\alpha$ implies that
\begin{equation}
    R^s_\KI\left( \Lambda\right) \leq \liminf_\alpha R^s_\KI\left( \Lambda_\alpha \right) .
\end{equation}
\end{lemma}

\begin{proof}
Due to Lemma~\ref{lem:channel_rob_subnormalised}, we can see that $2R^s_\KI+1$ is the gauge function (Minkowski functional) with respect to the set $\co \left( \wt\KI \cup \big(-\wt\KI\big)\right)$, that is,
\begin{equation}
2R^s_\KI(\Lambda)+1 = \inf \lset \lambda \bar \Lambda \in \lambda \co \left( \wt\KI \cup \big(-\wt\KI\big)\right) \rset.
\label{channel_rob_gauge}
\end{equation}
Crucially, from the weak*-operator compactness of $\wt\KI$ established in Corollary~\ref{KE_tilde_compact_cor} we have that $\co \left( \wt\KI \cup \big(-\wt\KI\big)\right)$ is also weak*-operator compact, which in particular implies that it is weak*-operator closed. The proof is then completed by noting that the gauge of a closed set is always lower semicontinuous. More explicitly, lower semicontinuity of $2R^s_\KI+1$ is equivalent~\cite[Proposition~2.5]{BARBU} to the weak*-operator closedness of the sublevel sets
\begin{equation}\begin{aligned}
s_\lambda \coloneqq& \lset \Lambda \bar 2R^s_\KI(\Lambda)+1 \leq \lambda \rset\\
=& \lset \Lambda \bar \Lambda \in \lambda\, \co \left( \wt\KI \cup \big(-\wt\KI\big)\right) \rset\\
=& \lambda \co \left( \wt\KI \cup \big(-\wt\KI\big)\right)
\end{aligned}\end{equation}
for all $\lambda$, which is immediate from the closedness of $\co \left( \wt\KI \cup \big(-\wt\KI\big)\right)$.
\end{proof}

We then proceed by establishing the identity in Lemma~\ref{lem:rob_psi_ineq} for the case of finite-dimensional channels.

\begin{lemma}\label{lem:rob_channels_psi_finite}
For any point-to-point channel $\Lambda : A \to B$ where $d_A, d_B < \infty$, it holds that
\begin{equation}\begin{aligned}
    R^s_{\KI} ( \Lambda  ) = \max_{\rho} R^s_\K\left(\left[\idc \otimes \Lambda\right] (\rho)\right) .
\end{aligned}\end{equation}
\end{lemma}

\begin{proof}[Proof of Lemma~\ref{lem:rob_channels_psi_finite}]
We use $J_\Gamma \coloneqq [\idc \otimes \Gamma](\Phi_{d_A})$ to denote the Choi state of a channel $\Gamma : A \to B$. Recall that $\Gamma \in \KI$ if and only if $J_\Gamma \in \K$~\cite{HSR,horodecki_2000-1}. By Lemma~\ref{lem:channel_rob_subnormalised} we can write
\begin{equation}\begin{aligned}
    R^s_{\KI} ( \Lambda ) &= \min \lset \lambda \bar J_\Lambda + \lambda J_\Gamma = (1+\lambda) J_\Theta,\; J_\Theta, J_\Gamma \in \K,\; \Tr_B J_\Theta = \frac{\id}{d_A} \rset\\
    &= \min \lset \lambda \bar J_\Lambda \leq_{\K} (1+\lambda) J_\Theta,\; J_\Theta \in \K,\; \Tr_B J_\Theta = \frac{\id}{d_A} \rset\\
   &= \min \lset \lambda \bar J_\Lambda \leq_{\K} (1+\lambda) J_\Theta,\; J_\Theta \in \K,\; \Tr_B J_\Theta \leq \frac{\id}{d_A} \rset,
\end{aligned}\end{equation}
where $\leq_{\K}$ denotes inequality with respect to the cone $\K$, in the sense that $X \leq_{\K} Y \iff Y - X \in \K$. We then have
\begin{equation}\begin{aligned}
    R^s_{\KI} ( \Lambda ) + 1 &= \min \lset \lambda \bar J_\Lambda \leq_{\K} J_\Theta,\; J_\Theta \in \K,\; \Tr_B J_\Theta \leq \lambda \frac{\id}{d_A} \rset\\
    &= \min \lset d_A \left\|\Tr_B J_\Theta\right\|_{\infty} \bar J_\Lambda \leq_{\K} J_\Theta,\; J_\Theta \in \K \rset\\
    &= \min_{\substack{ J_\Lambda \leq_{\K} J_\Theta \\ J_\Theta \in \K}} \max_{\rho_A} d_A \Tr \left[ (\rho_A \otimes \id) J_\Theta \right]\\
    &= \max_{\rho_A} \min_{\substack{ J_\Lambda \leq_{\K} J_\Theta \\ J_\Theta \in \K}} d_A \Tr \left[ (\rho_A \otimes \id) J_\Theta \right]
\end{aligned}\end{equation}
by Sion's minimax theorem~\cite{Sion}\footnote{\tcb{Sion's theorem gives us sufficient conditions for a function $f:\pazocal{X}\times \pazocal{Y}\to \R$ to satisfy the `minimax' property
\bb
\sup_{y\in \pazocal{Y}} \inf_{x\in \pazocal{X}} f(x,y) = \inf_{x\in \pazocal{X}} \sup_{y\in \pazocal{Y}} f(x,y)\, .
\ee
A set of conditions under which the above equality holds is as follows: (i)~$\pazocal{X}$ is compact and convex; (ii)~$\pazocal{Y}$ is convex; (iii)~$f(\cdot, y)$ is convex and lower semi-continuous on $\pazocal{X}$ for every $y\in \pazocal{Y}$; and (iv)~$f(x,\cdot)$ is concave and upper semi-continuous on $\pazocal{Y}$ for every $x\in \pazocal{X}$. In our case, $f$ is actually a bilinear function on a finite-dimensional space, hence verifying the above conditions (i)--(iv) is straightforward.}}. The rest of the proof will follow an argument similar to~\cite[Lemma~7]{berta_2017}. By continuity, it suffices to consider $\rho_A > 0$. Since conjugation by a product operator preserves $\K$-ness (that is, it cannot map a separable / PPT operator to an operator which is not separable / PPT, respectively) we have that
\begin{equation}\begin{aligned}
    J_\Lambda \leq_\K J_\Theta &\iff (\rho_A^{1/2} \otimes \id) J_\Lambda (\rho_A^{1/2} \otimes \id) \leq_\K (\rho_A^{1/2} \otimes \id) J_\Theta (\rho_A^{1/2} \otimes \id).
\end{aligned}\end{equation}
Defining $J_{\Theta'} \coloneqq (\rho_A^{1/2} \otimes \id) J_\Theta (\rho_A^{1/2} \otimes \id)$, we similarly have that $J_{\Theta'} \in \K \iff J_{\Theta} \in \K$. Altogether, this gives
\begin{equation}\begin{aligned}
    R^s_{\KI} ( \Lambda ) + 1 &= \sup_{\rho_A > 0} \min \lset  d_A \Tr J_{\Theta'} \bar (\rho_A^{1/2} \otimes \id) J_\Lambda (\rho_A^{1/2} \otimes \id) \leq_{\K} J_{\Theta'}, \; J_{\Theta'} \in \K \rset\\
    &= \sup_{\rho_A > 0} \min \lset d_A \Tr J_{\Theta'} \bar \idc \otimes \Lambda \left[ (\rho_A^{1/2} \otimes \id) \Phi (\rho_A^{1/2} \otimes \id) \right] \leq_{\K} J_{\Theta'}, \; J_{\Theta'} \in \K \rset\\
    &= \sup_{\rho_A > 0} R^s_\K\left(\idc \otimes \Lambda \left[ (\rho_A^{1/2} \otimes \id) d_A \Phi (\rho_A^{1/2} \otimes \id) \right] \right) +1\, ,
\end{aligned}\end{equation}
where $\Phi$ denotes the maximally entangled state. Since any pure state $\psi_{AA}$ can (up to an inconsequential local unitary on the second system) be written as $(\rho_A^{1/2} \otimes \id) d_A \Phi (\rho_A^{1/2} \otimes \id)$ for some $\rho_A$ and, conversely, any state $\rho_A$ can be purified to a state $\psi_{AA}$, we get
\begin{equation}\begin{aligned}
R^s_{\KI} ( \Lambda ) + 1 &= \max_{\psi \in \D(\HH_A \otimes \HH_A)} R^s_\K\left([\idc \otimes \Lambda](\psi)\right).
\end{aligned}\end{equation}
The proof is concluded by noting that the convexity of $\K$ ensures that $R^s_\K\left([\idc \otimes \Lambda](\psi)\right)$ is convex in $\psi$, which means that we can equivalently optimise over all states $\rho \in \D(\HH_A \otimes \HH_A)$ as the maximum will anyway be achieved at an extreme point $\psi$.
\end{proof}

The final step is to extend this relation to infinite-dimensional spaces. The first part of the proof is a standard argument based on finite-dimensional approximations of infinite-dimensional channels~\cite{shirokov_2008}, where we employ in particular a normalised, trace-preserving construction found e.g.\ in~\cite{Alex2017} in order to avoid normalisation issues. The second part of the proof relies on the lower semicontinuity that we have shown in Lemma~\ref{channel_robustness_lsc_lemma}.

\begin{proof}[Proof of Lemma~\ref{lem:rob_psi_ineq2}]

Let $\{\Pi_k\}_{k\in \N}$ and $\{\Pi'_k\}_{k\in \N}$ be increasing sequences of finite-rank orthogonal projectors which converge strongly to the identity operator on $\HH_A$ and $\HH_B$, respectively. For any channel $\Lambda : A \to B$, we define the maps
\begin{equation}\begin{aligned}\label{eq:lambda_k}
    \Lambda_k(X) \coloneqq \Pi'_k \Lambda ( \Pi_k X \Pi_k) \Pi'_k + \Tr \left[(\id - \Pi'_k)\, \Lambda(\Pi_k X \Pi_k) \right] \omega,
\end{aligned}\end{equation}
where $\omega \in \T(\HH_B)$ is some fixed state %chosen arbitrarily.
satisfying $\omega \leq \Pi'_k$ for all sufficiently large $k$, but otherwise arbitrary. Equivalently, this means that $\supp(\omega)\subseteq \supp(\Pi'_k)$ for some $k$, and hence for all $k'\geq k$.
It is not difficult to see~\cite{Alex2017} that the maps $\Lambda_k$ converge to $\Lambda$ in the topology of strong convergence --- specifically, for any $X \in \T(\HH_A)$, it holds that
\begin{equation}\begin{aligned}
    \lim_{k \to \infty} \left\| \Lambda_k(X) - \Lambda(X) \right\|_1 = 0\, .
    \label{Lambda_k_strong_operator_convergence}
\end{aligned}\end{equation}
Our strategy will now be to show that
\begin{equation}\begin{aligned}
    \limsup_{k \to \infty} \sup_{\rho\, =\, \id\otimes \Pi_k\, \rho\, \id\otimes \Pi_k} R^s_\K\left(\left[\idc \otimes \Lambda_k\right] (\rho)\right) \textleq{(i)} \sup_{\rho} R^s_\K\left(\left[\idc \otimes \Lambda\right] (\rho)\right) \textleq{(ii)} R^s_{\KI} (\Lambda) \textleq{(iii)} \liminf_{k \to \infty} R^s_{\KI} (\Lambda_k) ,
\end{aligned}\end{equation}
and use the finite-dimensional result of Lemma~\ref{lem:rob_channels_psi_finite} to conclude that equality holds between the leftmost and rightmost terms, since each $\Lambda_k$ can be equivalently understood as a map between finite-dimensional spaces.

We begin with the leftmost inequality~(i). Clearly, constraining to finite-dimensional input states $\rho$ such that %$\rho = \Pi_k \rho \Pi_k$
$\rho = \id\otimes \Pi_k\, \rho\, \id\otimes \Pi_k$ can only decrease the value of $\sup_{\rho} R^s_\K\left(\left[\idc \otimes \Lambda\right](\rho)\right)$. For any such state $\rho$, consider then any feasible solution for $R^s_\K\left(\left[\idc \otimes \Lambda\right](\rho)\right)$, that is, states $\sigma, \sigma' \in \K$ such that $\left[\idc \otimes \Lambda\right](\rho) + \lambda \sigma = (1+\lambda) \sigma'$. Note then that, since $\rho =  \id\otimes \Pi_k\, \rho\, \id\otimes \Pi_k$, it holds that $\idc \otimes \Lambda_k(\rho) = (\idc \otimes \Phi)\circ(\idc \otimes \Lambda)(\rho)$, where $\Phi(X) \coloneqq \Pi'_k X \Pi'_k + \Tr\left[(\id - \Pi'_k)X\right] \omega$. The crucial observation is that $\idc \otimes \Phi$ is a $\K$-preserving channel: simply projecting with a local projection $\id \otimes \Pi'_k$ cannot generate entanglement or non-positive partial transpose, and the measure-and-prepare map $X \mapsto \Tr\left[(\id - \Pi'_k)X\right] \omega$ is $\K$-enforcing, so by convexity of $\K$ we have that $\sigma \in \K \Rightarrow \idc \otimes \Phi(\sigma) \in \K$. This gives
\begin{equation}\begin{aligned}
    [\idc \otimes \Lambda_k](\rho) &= (1+\lambda) \underbrace{[\idc \otimes \Phi](\sigma')}_{ \in \K_{AB}} - \lambda \underbrace{ [\idc \otimes \Phi] (\sigma)}_{ \in \K_{AB}},
\end{aligned}\end{equation}
which constitutes a feasible solution for the robustness of $\idc \otimes \Lambda_k(\rho)$. Thus
\begin{equation}\begin{aligned}
    \sup_\rho R^s_\K\big([\idc \otimes \Lambda](\rho)\big) &\geq \sup_{\rho = \id \otimes \Pi_k \rho \id \otimes \Pi_k} R^s_\K\big([\idc \otimes \Lambda](\rho)\big)\\
    &\geq \sup_{\rho = \id \otimes \Pi_k \rho \id \otimes \Pi_k} R^s_\K\big([\idc \otimes \Lambda_k](\rho)\big)
\end{aligned}\end{equation}
for any $k$, from which inequality~(i) follows.

We now move on to inequality~(ii). 
Consider any feasible solution for $R^s_{\KI} ( \Lambda )$, that is, take any pair of channels $\Gamma, \Theta \in \KI$ such that $\Lambda + \lambda \Gamma = (1+\lambda) \Theta$. Then, for any $\rho$ we have that
\begin{equation}\begin{aligned}
    \left[\idc \otimes \Lambda\right] (\rho) &= (1+\lambda) [\idc \otimes \Theta](\rho) - \lambda [\idc \otimes \Gamma](\rho)\\
    &= (1+\lambda) \sigma - \lambda \sigma'
\end{aligned}\end{equation}
for some $\sigma, \sigma' \in \K^1$ by definition of $\KI$. This then gives $R^s_\K([\idc \otimes \Lambda] (\rho)) \leq \lambda$.
As this holds for any input state $\rho$ and any feasible $\lambda$, we get the desired inequality.

Finally, inequality~(iii) is just the lower semicontinuity of $R^s_\KI$ established in Lemma~\ref{channel_robustness_lsc_lemma}. To see that this is applicable, observe that the strong operator convergence in~\eqref{Lambda_k_strong_operator_convergence} implies in particular that $\Lambda_k \tends{w*o}{k\to\infty} \Lambda$. This concludes the proof.
\end{proof}

\begin{rem}
All of the considerations of this section, and in particular the main result of Lemma~\ref{lem:rob_psi_ineq2}, can be analogously applied to another resource measure closely related to the robustness $R^s_\KI$: the \emph{generalised robustness} $R^g_\KI$, defined as
\begin{equation}\begin{aligned}
    R^g_{\KI} ( \Lambda ) \coloneqq \inf \left\{ \lambda \,:\, \Lambda + \lambda \Gamma \in (1+\lambda) \KI,\; \Gamma \in \mathrm{CPTP}  \right\},
\end{aligned}\end{equation}
where now $\Gamma$ is not required to be a $\K$-enforcing channel. 
Indeed, the finite-dimensional variant of this result (analogous to our Lemma~\ref{lem:rob_channels_psi_finite}) appeared already in~\cite[Lemma~17]{regula_2020-2}. An extension of this finding to infinite-dimensional spaces, including a proof that $R^g_\KI$ is weak*-operator lower semicontinuous, can be obtained in direct analogy with our Lemmas~\ref{channel_robustness_lsc_lemma} and \ref{lem:rob_psi_ineq2}.
\end{rem}

%\bigskip
\section*{acknowledgments}

We thank Marco Tomamichel for enlightening discussions. L.L.\ was supported by the Alexander von Humboldt Foundation. B.R.\ was supported by the Japan Society for the Promotion of Science (JSPS) KAKENHI Grant No.\ 21F21015 and the JSPS Postdoctoral Fellowship for Research in Japan.

\bibliographystyle{apsrev4-1a}
\bibliography{biblio}

\end{document}